\documentclass[jair,twoside,11pt,theapa]{article}

\usepackage{hyperref}
\usepackage{jair, theapa, rawfonts}

\usepackage{relsize}
\usepackage{scalerel}
\usepackage[title]{appendix}
\usepackage{comment}

\ShortHeadings{Worst-case Bounds on Power vs. Proportion in Weighted Voting Games}{Gafni, Lavi \& Tennenholtz}
\firstpageno{1}

\usepackage{csquotes}
\usepackage{mathtools}
\usepackage{amsmath}
\usepackage{amsthm}
\usepackage{bm}
\usepackage{bbm}
\usepackage{amssymb}
\usepackage[]{algorithm2e}
\usepackage{graphicx}
\usepackage{amsmath}
\usepackage[cmyk]{xcolor}
\selectcolormodel{cmyk}

\usepackage{thmtools}
\usepackage{thm-restate}
\usepackage[]{algorithm2e}
\usepackage{pgfplots}

\newtheorem{lemma}{Lemma}[section]
\newtheorem{claim}{Claim}[section]

\newtheorem{theorem}{Theorem}[section]
\newtheorem{conjecture}{Conjecture}[section]
\newtheorem{remark}{Remark}[section]
\newtheorem{corollary}{Corollary}[section]

\newtheorem{definition}{Definition}[section]
\newtheorem{example}{Example}[section]
\allowdisplaybreaks

\usepackage{setspace}

\title{Worst-case Bounds on Power vs. Proportion in Weighted Voting Games with Application to False-name Manipulation}

\date{}

\author{\name Yotam Gafni \email yotam.gafni@campus.technion.ac.il \\
\addr Technion - Israel Institute of Technology \\
Haifa 32000 Israel
       \AND
       \name Ron Lavi \email ronlavi@ie.technion.ac.il \\
       \addr Technion - Israel Institute of Technology \\
Haifa 32000 Israel \\
University of Bath, UK \\
Bath BA2 7AY, United Kingdom
       \AND
       \name Moshe Tennenholtz \email moshet@ie.technion.ac.il \\
       \addr Technion - Israel Institute of Technology \\
Haifa 32000 Israel}

 \pgfplotsset{compat=1.14}
\begin{document}

 
\maketitle


\begin{abstract}
Weighted voting games apply to a wide variety of multi-agent settings. They enable the formalization of power indices which
quantify the coalitional power of players.
We take a novel approach to the study of the power of big vs.~small players in these games. We model small (big) players as having single (multiple) votes. The aggregate relative power of big players is measured w.r.t.~their votes proportion. 
For this ratio, we show small constant worst-case bounds for the Shapley-Shubik and the Deegan-Packel indices. In sharp contrast, this ratio is unbounded for the Banzhaf index.
As an application, we define a false-name strategic normal form game where each big player may split its votes between false identities, and study its various properties.
Together, our results provide foundations for the implications of players' size, modeled as their ability to split, on their relative power.
\end{abstract}


\section{Introduction}
Weighted Voting Games (WVGs) are a class of cooperative games that naturally appear in diverse settings, such as parliaments, councils, and firm shareholders. In recent years, they were found to naturally appear in multi-agent systems such as 
VCG auctions \cite{bachrachVCG} and other online economic systems. WVGs are defined by a set of players, their weights, and a threshold $T$. A set of players forming a coalition must have an aggregate weight of at least $T$. It is natural to ask: What is a player's power to influence decisions, or, alternatively, what is a player's share of the benefit of forming a coalition? This power measure does not necessarily comply with the player's proportional weight. For example, consider a WVG with a large threshold $T$, a big player with weight $T-1$ and a small player with weight $1$. Despite the large discrepancy in their weights, a consensus is required for any motion to pass, suggesting they have equal power.
This view of power considers a player's pivotal role as ``king-maker''
-- ``To the victors go the spoils''. 

Due to this reason, cooperative game theory studies power indices to capture 
the true effective power of players in WVGs.
This literature views the power index of a player as a numeric predictor of utility.
The most prominent power indices include the Shapley-Shubik index \cite{shapleyshubik}, which stems from the more general Shapley value \cite{shapley1952value}, and the Banzhaf index \cite{banzhaf1964weighted}. Other power indices emphasize different aspects of the power structure, such as the Deegan-Packel index \cite{deeganPackel}, which we also study. Our work lies in the intersection of three strands of WVG literature:

\subsubsection{Big vs.~Small Players and Group Power.}
We layout a WVG model with big vs.~small players and study the power of big players compared to their vote proportion, assuming that all player weights are natural numbers and consider all players with weight larger than 1 as ``big'' and all players with weight 1 ``small''.
%
The inequality in voting -- ``big vs.~small'' -- is a main drive for the study of power indices, going back to the formation of the US electoral college. \citeA{riker1986first} points out that
Luther Martin of Maryland, a staunch anti-federalist and one of the US founding fathers, analyzed the then-forming electoral college in a manner similar to the Banzhaf index. In the compilation by \citeA{antifederalist}, p.~50, Martin claims:
\begin{displayquote}
The number of delegates ought not to be in exact proportion to the number of inhabitants, because the influence and power of those states whose delegates are numerous, will be greater [even relative to their proportion] when compared to the influence and power of the other states...
\end{displayquote}
%
The contrast between big vs.~small players exists not only in traditional voting settings but also in modern contexts. For example, we see multiple situations in which several small websites (where here ``small'' is in terms of their number of users) aggregate their market power by forming a unified service platform to compete with a big incumbent website. Similarly, we see an aggregation of computational power (e.g., mining pools in Bitcoin, consortia of cloud computing services). The other direction of a big player splitting itself to multiple small identities also exists, even when such a split is costly in terms of advertising and maintaining the brand, e.g., flight search engines and web hosting services.

\citeA{shapleyshubik} demonstrate that in the settings of one big player and many small players, the power of the big player can be higher than its proportional weight. They do not answer (nor ask) the question of how large this ratio can be. They also do not analyze the opposite direction, of whether this ratio is bounded below by some constant, possibly smaller than 1.
As we explain below, our results generalize and answer these questions. Beyond the case of one big player and many small players, which we completely characterize, our model and results extend in two aspects:

\begin{itemize}

\item \emph{Arbitrary number of big (and small) players} - We obtain results regarding the power vs.~proportion of any specific one big player and regarding the {\em aggregate} power of big players. As the distinction between big and small immediately suggests, the relative power of these groups now becomes the focus.
Examples for such groups are the top 1\% wealthy people, the G7 countries, Bitcoin's large miners.
These settings typically involve a few big players and many small players. \citeA{oceanicgames} and \citeA{neyman1981singular} take this to an extreme by considering so-called ``oceanic games'' where there are a few significant large players and a continuum of small players. In contrast, our results are not asymptotic and hold for any arbitrary number of big and small players.

\item \emph{The Banzhaf and Deegan-Packel indices} - Since different power indices structuring naturally encapsulates different aspects of strategic power, it is important to compare the results of different indices given our model, and more so as the Banzhaf index gives qualitatively different results than the other indices.

\end{itemize}



\subsubsection{Power vs.~Proportion}
Our main theoretical results, given in Sections~\ref{powerVsProportionShapley}-\ref{powerVsProportionDP}, characterize the ratio between the aggregate
power of the big players and their aggregate proportional weight for different power indices.
Most previous literature analyzes ways to adjust voters' voting weights in order to equate voting power to actual weight.
An early suggestion by \citeA{penrose1946elementary} is that, in the UN, states should be assigned seats proportional to the square root of their population, in order to achieve proportional representation for each citizen, worldwide, regardless of her state. \citeauthor{Somczyski2006PenroseVS} \citeyear{Somczyski2006PenroseVS,DoubleSquareRootVotingSystem} further suggest an improvement in the form of the \emph{double square root} voting system, where on top of assigning seats proportional to the square root, the voting threshold (quota) of the representative body itself is determined so to optimize proportionality. 

More recent attention to whether a good choice of the quota can attain proportionality is found in the works of  \citeA{zick2011quota,zick2013random,randomizedWVGQuota2014,randomizedWVGCharacterization2016} and \citeA{superIncreasingWeights2016}. Theoretical guarantees, experimental results, and probabilistic models are suggested, for which this occurs. For example, these works collectively establish that under some probabilistic assumptions, setting the threshold $T$ to be about 50\% of the total sum of weights results in power being equal to proportion with high probability. In contrast, our worst-case analysis of this problem does not depend on probabilistic assumptions that might not hold in reality, due to the independence assumptions or specific properties of the distributions. We also show examples, where the threshold is very close to 50\%, and the power is far from proportional. In addition, it may not be possible to tune the threshold $T$ because of exogenous dictates (e.g., important parliament votes, where a two-thirds majority is required) or because the model aims to capture some underlying reality (e.g., over the Internet) that constrains $T$.

A third approach focuses on the probabilistic modelling of the WVG weights. 
For example,
\citeA{jelnov2014voting} show that if player weights are sampled uniformly from the unit simplex, the expected Shapley-Shubik power of a player relative to its proportion goes to 1 with rapid convergence in the number of players. \citeA{PenroseLimitTheorem} study a different model where the ratio of the Shapley-Shubik index to proportional weight in infinite chains of game instances asymptotically approaches 1.
\citeA{PenroseLimitTheoremExperimental}  follow up with an experimental analysis of a similar model thus further verifying the previous conceptual conclusions.


\subsubsection{False-name manipulation in WVG}
As an application of our main results for power vs.~proportion, we define in Section \ref{FalseNameGameSection} a false-name strategic normal-form game where each big player may split its votes between false identities.
Power indices in weighted voting games have different aspects: They can either measure the level of \emph{influence} of the players over possible decisions, or they can measure what \citeA{felsenthalMachoverPpower} call the ``P-power'': The ability to extract rent, in terms of monetary transfers for their causes. In parliamentary politics, this may take the form of budget or regulation dedicated to special interest groups. \citeA{peytonYoung} emphasizes that a goal of power indices measurement is to be a measure of how much ``bribe'' a player can extract for her vote. This view is important to justify why false-name attacks are interesting and effective in WVGs: While in terms of influence, splitting into multiple players may not matter (as they still vote ideologically as a bloc), in terms of rent extraction, this restructuring of party landscape matters. 

\citeA{aziz2011false} are the first to study power indices in the context of false-name manipulation, showing upper and lower bounds on a player's gain (or loss) from splitting its votes into two parts, for the Shapley-Shubik and Banzhaf indices. They also address a range of computational issues, among them the decision problem of splitting into two equal parts, which is NP-hard for both indices. \citeA{comparisonComplexity} show that the decision problem of benefiting from splitting into two equal parts to be in PP, and \citeA{probPolytime} show it is PP-complete for the Shapley-Shubik index, PP-complete for the Banzhaf index with three equal splits, and PP-hard for both indices with general splits.

Our results contribute to the above literature on false-name splits by unifying it with the two previously mentioned aspects (`big vs.~small' and `power vs.~proportion') and by generalizing on two additional fronts:

\begin{itemize}

    \item \emph{General splits} - We consider splits into multiple identities, rather than splits into two or three identities.
 \citeA{lasisi2017false} initiated work on this more general problem, where they show some upper and lower bounds for the individual power gain from general splits compared to the original power. These bounds assume that only a single agent splits, whereas the bounds we provide hold under any combination of strategic manipulations by the agents.

    \item \emph{Global bounds on manipulation} - By our results for the power vs.~proportion we extract useful global bounds on manipulation. Regularly power is compared before and after splits. Since previous work shows the most basic questions in regard to successful power manipulation to be computationally hard, developing global performance bounds is important.

\end{itemize}

\subsection{Overview of Main Results}

\begin{itemize}
    \item 
    Our main theoretical results, given in Sections~\ref{powerVsProportionShapley}-\ref{powerVsProportionDP}, characterize the ratio between the aggregate power of the big players and their aggregate proportional weight. For the Shapley-Shubik index, this ratio is bounded by 2, and this analysis is asymptotically tight. In contrast, for the Banzhaf index, this ratio is unbounded. The Deegan-Packel index gives the same qualitative result as the Shapley-Shubik index: For this index, we prove an upper bound of $3$ for the ratio in the worst case and give an example where the ratio is asymptotically at least 2.
    
    We give additional examples showing that while these results hold for the aggregate power of all big players, they do not hold for the power of any individual player. In contrast to these upper bounds, we show that the power to proportion ratio cannot be lower bounded by any positive number. Thus, while the aggregate power of the big group cannot be much larger than its proportion, it can certainly be much smaller.
    
    \item 
    As an application of these results, given in Section~\ref{powerVsProportion},
    Section~\ref{FalseNameGameSection} defines a false-name strategic normal-form game where each big player may split its votes between false identities. Subsection~\ref{subsection:individualResults} characterizes the worst-case effects of splits on the individual power of a strategic big player in the false-name game. In the special case of one big player,
    if the big player splits its vote tokens to $k$ false identities, each with a weight of 1, its power will be equal to its proportional weight (since all power indices are symmetric). 
    By our results, this is never less than one-half of the power of the optimal split for the Shapley-Shubik index.
    On the other hand, if the big player decides not to split, in some cases its power is unboundedly small relative to its proportional weight.
    
    \item
    To generalize the above conclusion, derived for the case of one big player, to the aggregate power of any number of big players and any possible set of players' split strategies, we give a combinatorial conjecture for the Shapley-Shubik index. We support our conjecture by extensive experimentation. The conjecture claims that the big players may lose at most half of their aggregate power from any combination of splits.
    Overall the results can be interpreted in the following way: As a whole, the big players are not set to lose much from opening the voting game to false-name manipulations, while on the other hand some notion of fairness can be guaranteed as they can not exceed double their proportional size. 
\end{itemize}

\section{Preliminaries}
\label{powerVsProportion}

Weighted Voting Games (WVGs), starting with weighted majority games 
\cite{von1947theory,shapley1962simple}, 
aim to capture a situation where several players need to form a coalition. Each player has a weight, and a subset of players can form a coalition if their sum of weights passes a certain threshold. In this paper, we make a distinction between ``big'' versus ``small'' players, where small players have a weight of one. Formally,

\begin{definition} (Adapted from \citeR{shapley1962simple})
A weighted voting game is a tuple $\{A,m,T\}$ with 
\[
\begin{split}
& A = \{a_1,...,a_r\}, M = \overbrace{\{1,...,1\}}^m,  \quad 1\leq T \leq m + \sum\limits_{j=1}^r a_j,
\end{split}
\]
\noindent where $a_1,...,a_r,m,T \in \mathbb{N}$, there are $r$ ``big players'', $m$ ``small players'' of weight 1, and a coalition threshold $T$.\footnote{There is some loss of generality by fixing the parameters $a_i,T$ to be exact multiples of the weight of the  small  player. Our model can be slightly generalized as follows:
 Let $s$ be some minimal weight corresponding to some operational or regulatory minimal size of a venture, or to an electoral threshold for parliaments. Any  player  with  an  integer  weight $s \leq w < 2s$ is termed ``small'' as small players are the ones that  cannot  split. The model as presented corresponds to the case $s=1$ for tractability and readability but we believe that our results hold for the more general model as well. We give more details in the discussion section.} We at times denote the small players as $1_1,...1_m$. Note that $A$ is a multi-set. When we write $A \setminus \{i\}$, for some weight $i$, at most one occurrence of $i$ is removed from $A$.
\end{definition}

A basic question in WVGs is how to split the gains from forming a coalition among its members. One possible notion of fairness is to split gains in a way that is approximately proportional to the weights of the coalition members.

\begin{definition}
The proportional value of a weighted voting game is $\displaystyle P(A,m,T) = \frac{\sum\limits_{j=1}^r a_j}{m + \sum\limits_{j=1}^r a_j}.$
\end{definition}

However, reality tells us that many times the ``power'' of players is different than their proportional weight. Well-established literature on power indices formally studies this by looking at our setting as a cooperative game.
For the analysis we have the following value function $v(S)$ which describes whether a subset of players $S\subseteq A\cup M$ is able to form a coalition:
$$v(S) = \begin{cases} 
1 & \sum\limits_{s\in S}s \geq T \\
0 & otherwise.
\end{cases}$$
We next compare the aggregate power of the big players, using several three well-known power indices, to their aggregate proportional weight. 

\section{The Shapley-Shubik Power Index}
\label{powerVsProportionShapley}

Define the ordered tuple $A+M = (a_1,...,a_r,1_1,...,1_m)$. 
Let $S_{m+r}$ be the group of all permutations operating on $m+r$ objects. For some $\sigma \in S_{m+r}$,  $\sigma(A+M)$ is the ordered tuple which results by applying the permutation $\sigma$ to $A+M$. We usually omit the term $A+M$ when it is clear from the context. Define $\displaystyle \sigma|_p, \bar{\sigma}|_p$ as the set of all players (strictly, non-strictly) preceding player $p$ in permutation $\sigma(A+M)$. 
The permutation pivotal player indicator function for a player $p$ is 
$$\displaystyle \mathbbm{1}_{p,\sigma} = v(\bar{\sigma}|_p) - v(\sigma|_p).$$
In words, the indicator $\mathbbm{1}_{p,\sigma}$ is equal to one if the players preceding $p$ in the permutation $\sigma(A+M)$ do not form a coalition and adding $p$ enables the coalition formation. In such a case, we say that $p$ is pivotal for $\sigma$.
%
%
%
Note that $\mathbbm{1}_{p,\sigma} \in \{0,1\}$ and that each permutation has exactly one pivotal player.

\begin{definition} (Adapted from \citeR{handbookGameTheory})
The Shapley-Shubik power index of a weighted voting game (A,m,T) is 
\[
\begin{split}
& \phi_{p}(A,m,T) = \mathbb{E}_{\sigma \sim UNI(S_{m+r})}[\mathbbm{1}_{p,\sigma}], \\
\end{split}
\]
for a player $p$ (whether a big player $a_i$ or a small player $1_i$), where UNI is the uniform distribution over a discrete set.
\end{definition}

\noindent
This definition is a special case of the Shapley value applied to WVGs. Three well-known properties of this power index are:

\begin{itemize}

\item {\bf Symmetry. } $a_i = a_j \rightarrow \phi_{a_i}(A,m,T) = \phi_{a_j}(A,m,T)$.
This implies that for any two small players $i,j$, $\phi_{1_i}(A,m,T) = \phi_{1_j}(A,m,T)$. Thus, for simplicity, we sometimes write $\phi_1(A,m,T)$.

\item {\bf Efficiency. } $\sum\limits_{j=1}^r \phi_{a_j}(A,m,T) + m \phi_{1}(A,m,T) = 1$, i.e., the sum of all players' Shapley-Shubik indices is 1. 

\item {\bf Non-Negativity. } The Shapley-Shubik index of any player $p$ is non-negative: $\phi_{p}(A,m,T) \geq 0$. 


\end{itemize}

\begin{definition}
The Shapley-proportional ratios are the global supremum (infiumum) over all weighted voting games
$$\bar{R}_{\phi}\!=\!\sup_{A,m,T} \frac{\sum\limits_{j=1}^r \phi_{a_j}(A,m,T)}{P(A,m,T)}, \underbar{R}_{\phi}\!=\!\inf_{A,m,T} \frac{\sum\limits_{j=1}^r \phi_{a_j}(A,m,T)}{P(A,m,T)}.$$
Note that $\underbar{R}_{\phi} \geq 0$ because of non-negativity. 
\end{definition}

\begin{example}[$\bar{R}_{\phi}$ is at least 2]
\label{lowerBoundShapley}
For some $k\geq 2$, consider $A = \{k\}, m = k-1, T = k$. Then
$\phi_{a_1}(A,m,T) = 1$, while $P(A,m,T) = \frac{1}{2} + \frac{1}{4k - 2}$.
\end{example}
\noindent
In fact, this asymptotic lower bound is tight:
\begin{theorem}
\label{fairnessTheorem}
$\bar{R}_{\phi} = 2.$
\end{theorem}


To prove Theorem~\ref{fairnessTheorem}, we first show a recursive relation for the Shapley-Shubik index. 

\begin{restatable}{lemma}{recursionLemma}
\label{recursionLemma}
The following recursion holds for $\phi_1$

\[
\begin{split}
 & \phi_1(A,m,T) = \begin{dcases}
  \frac{1}{m+r} &  T = 1 \\
 \frac{1}{m+r} \bigg(\sum\limits_{\substack{1\leq i \leq r \\ a_i < T}} \phi_1(A\!\setminus\!\{a_i\},m,T-a_i) + (m\!-\!1) \phi_1(A,m\!-\!1,T-1)\bigg) & T > 1
 \end{dcases}
\end{split}
\]

\end{restatable}

\begin{proof}
Fix a small player for which we measure the expected number of permutations where it is pivotal. The recursion is done by the conditional expectation on the identity of the first player in the permutation. If $T=1$, then if and only if our fixed player is first, is it pivotal, which happens in probability $\frac{1}{m+r}$. For $T>1$, if the first player in the permutation is $a_i$ and $a_i \geq T$, then our fixed player is not pivotal. If it is some big player with $a_i < T$, then the problem reduces to the WVG with parameters $(A \setminus \{a_i\}, m, T - a_i)$. Similarly, if a small player different than our fixed player is first (for which there are $m-1$ alternatives), the problem reduces to the WVG with parameters $(A, m-1, T-1)$. If our fixed player is first, it is not pivotal. All the above events are with a uniform probability of $\frac{1}{m+r}$. 
\end{proof}

We now show a sufficient condition to prove Theorem~\ref{fairnessTheorem}. Let $\displaystyle r_j = \left|\{i\}_{1\leq i \leq r, a_i = j}\right|$, the amount of big players with weight equal to $j$, and let $a = \max A$. 
We have by definition 
\begin{equation}
\label{defHistograms}
\begin{split}
    & r = \sum\limits_{j=2}^{a} r_j, \quad\quad \sum\limits_{j=1}^r a_j = \sum\limits_{j=2}^a r_j \cdot j.
    \end{split}
\end{equation}

\begin{restatable}{lemma}{shapleyThmInductionStatement}
\label{shapleyThmInductionStatement}
$\bar{R}_{\phi} = 2$, if for any WVG parameters $(A,m,T)$ it holds that $$\displaystyle \phi_1(A,m,T) \geq \frac{m-\mathlarger{\sum}\limits_{i=2}^{a} r_i \min(i, T)}{m\left(m+\mathlarger{\sum}\limits_{i=2}^{a} r_i \min(i, T)\right)}.$$ 
\end{restatable}
\begin{proof}

By Example~\ref{lowerBoundShapley}, $\bar{R}_{\phi} \geq 2$. We thus show $\bar{R}_{\phi} \leq 2$. 
By the premise of the lemma, 
\begin{equation}
\label{strongerInduction}
\begin{split}
& \phi_1(A,m,T) \geq \frac{m-\mathlarger{\sum}\limits_{i=2}^{a} r_i \min(i, T)}{m\left(m+\mathlarger{\sum}\limits_{i=2}^{a} r_i \min(i, T)\right)} \geq \\
& \frac{m-\mathlarger{\sum}\limits_{i=2}^{a} r_i \cdot i}{m\left(m+\mathlarger{\sum}\limits_{i=2}^{a} r_i \cdot i\right)} \stackrel{Eq.~\ref{defHistograms}}{=} \frac{m-\mathlarger{\sum}\limits_{j=1}^r a_j}{m\left(m+\mathlarger{\sum}\limits_{j=1}^r a_j\right)}.
\end{split}
\end{equation}

Now,

\begin{equation}
\label{weakInduction}
\begin{split}
&  \sum_{j=1}^r \phi_{a_j}(A,m,T) \stackrel{\text{efficiency + symmetry}}{=} 1 - m \phi_1(A,m,T) \stackrel{Eq.~\ref{strongerInduction}}{\leq } \\
& 1 - \frac{\displaystyle m - \sum\limits_{j=1}^r a_j}{\displaystyle m + \sum\limits_{j=1}^r a_j} =  2\frac{\displaystyle \sum\limits_{j=1}^r a_j}{\displaystyle m + \sum\limits_{j=1}^r a_j} = 2P(A,m,T).
\end{split}
\end{equation}

If we divide both sides by $P(A,m,T)$, and since this holds for any parameters $(A,m,T)$, we get $\bar{R}_{\phi} \leq 2$. 

\end{proof}

We now prove Theorem~\ref{fairnessTheorem} by proving the sufficient condition of Lemma~\ref{shapleyThmInductionStatement} holds by induction. 

\begin{lemma}
\label{shapleyThmInductionProof}
For any WVG $(A,m,T)$, $$\displaystyle \phi_1(A,m,T) \geq \frac{m-\mathlarger{\sum}\limits_{i=2}^{a} r_i \min(i, T)}{m\left(m+\mathlarger{\sum}\limits_{i=2}^{a} r_i \min(i, T)\right)}.$$
\end{lemma}

\begin{proof}

We prove the lemma by induction on the threshold value $T$. 
For $T=1$, we have $\phi_1(A,m,T) = \frac{1}{m+r} \geq \frac{m-r}{m(m+r)}$, which satisfies the condition of Eq.~\ref{strongerInduction}, since $\mathlarger{\sum}\limits_{i=2}^{a} r_i \min(i,T) = \mathlarger{\sum}\limits_{i=2}^{a} r_i = r$. 
Now take a threshold value $T$, and assume for all lower threshold values and any $A, m$, the condition holds. We make several justified assumptions:
\begin{itemize}
\item Recall that we denoted $a = \max A$, the largest weight of a big player. We assume w.l.o.g.~$a \leq T$. If it is strictly higher than $T$, consider the game $A'$ where all $a_i > T$ values are set to $T$. Then $\phi_1(A,m,T) = \phi_1(A',m,T)$ because these players are pivotal iff they appear in the permutation without a preceding pivotal player, in both cases. For $A'$, 
$a' = \max A' = T, \displaystyle r'_j = \left|\{i\}_{1\leq i \leq r, a'_i = j}\right|$, and
\[
\begin{split}
    & \sum\limits_{i=2}^{a'} r'_i \min(i,T) = \sum\limits_{i=2}^{T} r_i \min(i,T) + \sum\limits_{i=T+1}^a r_i T = \sum\limits_{i=2}^{a} r_i \min(i,T).
    \end{split}
    \]
Thus, proving the induction assumption for $(A',m,T)$ also proves it for $(A,m,T)$. 
\item We assume $m > \mathlarger{\sum}\limits_{j=1}^r a_j$. Otherwise, we have by the non-negativity property of the Shapley-Shubik index, 
\[
\begin{split}
    & m\phi_1(A,m,T) \geq 0 \geq \frac{m\!-\!\sum\limits_{j=1}^r a_j}{m\!+\!\sum\limits_{j=1}^r a_j} \stackrel{(1)}{=}  \frac{m\!-\!\mathlarger{\sum}\limits_{i=2}^{a} r_i \min(i, T)}{m\!+\!\mathlarger{\sum}\limits_{i=2}^{a} r_i \min(i, T)},
    \end{split}
    \]
where $(1)$ is since $a\leq T$, and therefore for $1\leq i\leq a, \min(i,T) = i$. 
\item We assume $m\geq 2$, as otherwise, by the $m > \mathlarger{\sum}\limits_{j=1}^r a_j$ assumption, $A = \emptyset$ and $\phi_1(A,m,T) = \frac{1}{m}$, which is exactly the required induction hypothesis in this case.
\end{itemize}
We now develop the induction step. We can write
\begin{align*}
\tag{Induction Step}
\label{recursionStep1}
& m \phi_1(A,m,T) \stackrel{Lemma~\ref{recursionLemma}}{=} \\
 & m\!\cdot\!\frac{\sum\limits_{\substack{1\leq k \leq r \\ a_k < T}} \hspace{-2ex} \phi_1(A\!\setminus\!\{a_k\},m,T\!-\!a_k)\!+\!(m\!-\!1) \phi_1(A,m\!-\!1,T\!-\!1)}{m+r} = \\
 & m\!\cdot\!\frac{(m\!-\!1)\overbrace{\phi_1(A, m\!-\!1,T\!-\!1)}^{Exp1}\!+\!\sum\limits_{i=2}^{T-1} r_i  \overbrace{\phi_1(A\!\setminus\!\{i\}, m,T\!-\!i)}^{Exp2}}{m+r} \stackrel{(1)}{\geq} \\
& \frac{1}{m+r} \cdot \bigg(m\cdot \frac{m-1 - \mathlarger{\sum}\limits_{j=2}^T r_j \min(j, T-1)}{m-1+\displaystyle \sum\limits_{j=2}^T r_j\min(j, T-1)} + \\
& \mathlarger{\sum\limits_{i=2}^{T-1}} r_i \!\cdot\!\frac{m\!-\!\big(\hspace{-2ex}\mathlarger{\sum}\limits_{2\leq j\neq i \leq T} \hspace{-2ex}r_j \min(j, T-i)\big)\!-\!(r_i\!-\!1)\min(i,T\!-\!i)}{m\!+\!\big(\hspace{-2ex}\mathlarger{\sum}\limits_{2\leq j\neq i\leq T}\hspace{-2ex} r_j \min(j, T\!-\!i)\big)\!+\!(r_i\!-\!1)\min(i,T\!-\!i)} \bigg) \stackrel{(2)}{\geq} \\
& \frac{m  -  \displaystyle \sum\limits_{j=1}^r a_j }{m + \displaystyle \sum\limits_{j=1}^r a_j} \stackrel{Eq.~\ref{defHistograms}}{=} \frac{m - \mathlarger{\sum}\limits_{j=2}^a r_j \cdot j}{m + \mathlarger{\sum}\limits_{j=2}^a r_j \cdot j} \stackrel{(3)}{=} \frac{m - \mathlarger{\sum}\limits_{j=2}^a r_j \min(j, T)}{m + \mathlarger{\sum}\limits_{j=2}^a r_j \min(j, T)}, 
\end{align*}
where in $(1)$ we replace $Exp1$ and $Exp2$ by the induction hypothesis for lower values of $T$. For $Exp2$, note that in the game $(A\setminus \{i\}, m, T-i)$, for any weight $2\leq j \neq i \leq T$, the number of players with weight $j$ is $r_j$, but for the weight $i$ it is $r_i - 1$. In Lemma~\ref{shapleyInductionInequalities} of the appendix we show $(2)$, and $(3)$ is since $a\leq T$. The RHS is the inductive result required by Eq.~(\ref{strongerInduction}). This completes the induction step. 
\end{proof}

We conclude the section by examining two more illustrative examples, and providing an individual upper bound for a player's Shapley-Shubik index. 

The following example shows that $\underbar{R}_{\phi} = 0$. 


\begin{example} 
\label{reverseRatio}
%
For any $k \geq 2$, choose $A = \{k\}, m = k, T = 2k$. Then
$\phi_{a_1}(A,m,T) = \frac{1}{k+1}$ while $P(A,m,T) = \frac{1}{2}$.
\end{example}


In example~\ref{reverseRatio}, a big player has less power in terms of the Shapley-Shubik index than its proportional weight. In the next example the opposite holds:
\begin{example}
\label{singlePlayerFairnessExample}
For a single player, it may hold that its individual power to proportional weight ratio is unbounded:
Consider $A = \{2,k\}, m = 1, T = k + 3$. Then $\phi_{a_1}(A,m,T) = \frac{1}{3}$ while $ \frac{a_1}{m + \sum\limits_{j=1}^r a_j} = \frac{2}{k+3}$.
\end{example}
Nevertheless, there does exist an upper bound on the Shapley-Shubik index of any individual big player:
\begin{theorem}
\label{lemmaBasicUpperbound}
$\displaystyle \phi_{a_i}(A,m,T) \leq \frac{a_i}{m+r}.$
\end{theorem}
\begin{proof}
Let $S_n$ be the set of all permutations over $n$ items and $\sigma(A)$ for some ordered set $A$ and some $\sigma \in S_n$ is the ordered set after applying $\sigma$. For $j \in \{1,...,|A|\}$ let $\sigma_j(A)$ be the element in the $j$-th position of the ordered set $\sigma(A)$, and $\sigma|_{-j}(A)$ be the ordered set of size $|A| - 1$ obtained by omitting the element in position $j$ of $\sigma(A)$. 

Two ways to think about sampling permutations for the calculation of the Shapley-Shubik index uniformly at random are the following: First, we can take $S_{m+r}$, the set of all permutations that are possible to apply to $A+M$, and sample $\sigma \sim UNI(S_{m+r})$. Equivalently, for a fixed player $a_i$, we can choose its position in the permutation applied to $A+M$, namely $1 \leq j \leq m+r$, uniformly at random, and then sample a permutation over all other players $\sigma' \sim UNI(S_{m+r-1})$. We then denote $\sigma^*_{a_i}(j, \sigma')$ as the unique permutation that has $\sigma_{a_i}^*|_j(A+M) = a_i$ and $\sigma_{a_i}^*|_{-j}(A+M) = \sigma'(A \setminus \{a_i\} + M)$. Also recall $[m+r]$ is a short-hand notation for the set $\{1,...,m+r\}$. 
This enables rewriting $\phi_{a_i}(A,m,T)$ in terms of total expectation:
\[
\begin{split}
    & \phi_{a_i}(A,m,T) = \mathbb{E}_{\sigma \sim UNI(S_{m+r})}[\mathbbm{1}_{a_i,\sigma}] = \\
    & \mathbb{E}_{\sigma' \sim UNI(S_{m+r-1})}\left[\mathbb{E}_{j \sim UNI([m+r])}[\mathbbm{1}_{a_i,\sigma^*_{a_i}(j, \sigma')} \mid \sigma' ] \right] \stackrel{(1)}{\leq} \\
    & \mathbb{E}_{\sigma' \sim UNI(S_{m+r-1})}\left[\frac{a_i}{m+r}\right] = \frac{a_i}{m+r}\ \ .
    \end{split}
    \]
\noindent
The transition in $(1)$ is by the following argument. Let $\nu_j(\sigma')$ be the sum of the first $j$ elements in $\sigma'$. By definition,
$a_i$ is pivotal in $\sigma^*_{a_i}(j, \sigma')$ iff $\nu_j(\sigma') <T$ and $\nu_j(\sigma') + a_i \geq T$. Each choice of $j$ determines a unique integer value for $\nu_j(\sigma')$. I.e., there are overall $a_i$ integer values that have the property that their sum is smaller than $T$ and with the addition of $a_i$ their sum is at least $T$, namely $T-1,...,T-a_i$.
There are at most $a_i$ unique values of $j$ (and, as a result, of $\nu_j(\sigma')$)  that would result in $j$ being pivotal for $\sigma^*_{a_i}(j, \sigma')$. This is out of the total of $m+r$ possible values for $j$, uniformly chosen. So for any $\sigma'$ the inner expectation is bounded by $\frac{a_i}{m+r}$. 
\end{proof}

\section{The Banzhaf Index}
\label{powerVsProportionBanzhaf}

We show a contrary result for the Banzhaf index \cite{banzhaf1964weighted}, where an asymptotic example has an unbounded ratio of aggregate big players' power over their proportion.

\begin{definition} (See \citeR{dubey1979mathematical})
Let $P(S)$ be the power set of $S$, and $UNI$ be the uniform distribution over a discrete set. The Absolute Banzhaf index of a WVG is:
\[
\begin{split}
    & \beta'_{a_i}(A,m,T) = \mathbb{E}_{S \sim UNI(P(A\setminus \{a_i\}\cup M)}\left[v(S\cup \{a_i\}) - v(S)\right] \\
    & \beta'_1(A,m,T) = \mathbb{E}_{S \sim UNI(P(A\cup M\setminus \{1\})}\left[v(S\cup \{1\}) - v(S)\right].
    \end{split}
\]
The Normalized Banzhaf index is:
$$ \beta_{a_i} = \frac{\beta'_{a_i}(A,m,T)}{\displaystyle \sum\limits_{j=1}^r \beta'_{a_j} + m\beta'_1(A,m,T)}.$$
\end{definition}
While the Shapley-Shubik index gives equal probabilities to all permutations over players, the Banzhaf index gives equal probabilities to all subsets of players. The normalization is needed to achieve the efficiency property, where summation over the indices of all players sums up to exactly one. The absolute Banzhaf indices may sum to less or more than 1. For an individual absolute Banzhaf index, by definition 
\begin{equation}
\label{eqBanzhafIndividualUpperBound}
    \beta'_{a_i}(A,m,T) \leq 1
\end{equation}

\begin{definition}
The Banzhaf-proportional ratios are the global supremum over all WVGs
\[
\begin{split}
& \bar{R}_{\beta'} = \sup_{A,m,T} \frac{\displaystyle \sum\limits_{i=1}^r \beta'_{a_i}(A,m,T)}{P(A,m,T)}, \\
\end{split}
\]
A similar definition holds for $\bar{R}_{\beta}$.
\end{definition}

While the power of the big players cannot be much larger than their proportional weight according to the Shapley-Shubik index, the Banzhaf index gives a different result:

\begin{theorem}
\label{banzhafUnbounded}
$\bar{R}_{\beta'}, \bar{R}_{\beta}$ are unbounded. 
\end{theorem}

The proof of Theorem~\ref{banzhafUnbounded}, given in Appendix~\ref{banzhafUnbounded}, shows that the ratio in the following example goes to infinity with $k$.

\begin{example}
\label{example_for_unbounded_Banzhaf_ratio}
Consider $A = \{2k\}, m = k^{1.5}, T = \frac{2k + k^{1.5}}{2}$. Then for $k = 1600$, we have
$$
    P(A,m,T)\!=\!\frac{3200}{3200\!+\!1600^{1.5}}\!=\!\frac{1}{21},
    \beta_{a_1}(A,m,T)\!\approx\! \beta'_{a_1}(A,m,T)\!\approx\! 
    1. \\
$$
\end{example}

The intuition for the calculation is as follows. Since there is only one big player, and all other players have a weight of one, only the size of the subset of other players matters for the index. Choosing a subset $S$ of small identical players with uniform probability is like letting each small player participate in the chosen subset with probability $\frac{1}{2}$ (a Bernoulli trial). So, the size of the subset is sampled from the Binomial distribution with parameters $B(m, \frac{1}{2})$. 
Measure concentration properties of the symmetric binomial distribution around its mean imply that with high probability the size of the sampled set is close enough to $\frac{m}{2}$ so that the big player in the example is pivotal. Details of this calculation can be directly extracted from the argument in the proof of Theorem~\ref{banzhafUnbounded}.
The normalized Banzhaf index is for these parameters is:
        \[
        \begin{split}
    & \beta_{a_1}(A,m,T)\!\approx\!
    \frac{0.999999}{0.999999\!+\!1600^{1.5}\cdot 3.52795*10^{-33}}\!\approx\!1.
        \end{split}
        \]
        
    This yields $21$ for both the ratio of absolute Banzhaf to proportion and normalized Banzhaf to proportion. Notice that this example also serves as a counter-example for a possible adaptation of Theorem~\ref{lemmaBasicUpperbound} to the Banzhaf index, as we have $\beta_{a_1}(A,m,T) \approx 1 > \frac{3200}{64001} = \frac{a_1}{m+r}$. 



Example~\ref{example_for_unbounded_Banzhaf_ratio} is in the spirit of Section b of \citeA{penrose1946elementary} and the asymptotic results analysed in Section 7 of \citeA{dubey1979mathematical}. As far as we know, the exact bound that we derive along with its formal analysis are new.

Theorem 7 of \citeA{aziz2011false} states that if a single player splits her votes between exactly two identities, its power as measured by the Banzhaf index cannot increase by a factor larger than 2. In contrast, Theorem~\ref{banzhafUnbounded} above shows that, with general splits, a player might end up {\em decreasing} its power by an unbounded factor. This paints an overall non-favorable picture for false-name manipulations as measured by the Banzhaf index. 
%

\section{The Deegan-Packel Index}
\label{powerVsProportionDP}

\begin{definition}
An all-pivotal set in a WVG is a set $S$ such that for any player $s\in S$, $v(S) - v(S\setminus s) = 1$. Let the set of all-pivotal subsets be $AP$.
The Deegan-Packel index is:
$$\rho_{a_i}(A,m,T) = \mathbb{E}_{S \sim UNI(AP)}\left[\frac{\mathbbm{1}_{a_i \in S}}{|S|}\right],$$
where $UNI$ is the uniform disribution over a discrete set. 
\end{definition}

Thus, the Deegan-Packel index is similar to the Banzhaf index but takes into account only the all-pivotal subsets that are, in a sense, the minimal coalitions. It also considers the size of the coalition, so participation in a large coalition results in less ``power'' than being a part of a small coalition. One can verify that this index is efficient, i.e., the sum of indices of all players is always exactly one.

\begin{definition}
The Deegan-Packel-proportional ratio is the global supremum over all WVGs
\[
\begin{split}
& \bar{R}_{\rho} = \displaystyle \sup_{A,m,T} \frac{\sum\limits_{i=1}^r \rho_{a_i}(A,m,T)}{P(A,m,T)}. 
\end{split}
\]

\end{definition}

\begin{example}
$\bar{R}_{\rho} \geq 2$:
Let $\displaystyle A = \{k\}, m = k-1, T = k$.
Then
$\displaystyle \rho_{a_1}(A,m,T) = 1$, while $P(A,m,T) = \frac{k}{2k-1}$. 
\end{example}

\begin{theorem}
\label{thm:DPpropRatio}
$\bar{R}_{\rho} \leq 3$. 
\end{theorem}

The theorem follows from lemmas given in Appendices~\ref{appendix:DPFirstPartProof},\ref{appendix:DPsecondPartProof}. It exploits properties of the all-pivotal coalition in two threshold regimes: If the threshold is high, we show that enough small players must participate in an all-pivotal coalition, making the relative power of the big players in such a coalition small. If the threshold is low, we are able to use algebraic manipulations over binomials to derive the bound. 

Regarding the power of each individual player, the following example shows that Theorem~\ref{lemmaBasicUpperbound} can not be adapted for the Deegan-Packel index. 

\begin{example}
\label{ex:DPnoBasicLowerBound}
There exists $\{A,m,T\}$ such that $\rho_{a_1}(A,m,T) > \frac{a_1}{m+r}$. 

\noindent Consider $A = \{99, \overbrace{101, ..., 101}^{1000}\}, m = 100, T = 101$. Any all-pivotal coalition $S$ with $a_i \in S$ for $i\geq 2$ is a singleton. Any coalition $S$ with $\forall i\geq 2, a_i \not \in S$ must include $a_1$, otherwise the coalition does not reach the threshold $T$. For a coalition $S$ with $a_1 \in S$, $\forall i\geq 2, a_i \not \in S$ to be all-pivotal it must include exactly two small players. There are thus $\binom{m}{2} = 4950$ all-pivotal coalitions that do not include any $a_i$ with $i\geq 2$, and $5950$ all-pivotal coalitions overall. We conclude that $\rho_{a_1}(A,m,T) = \frac{4950}{5950} \frac{1}{3} > \frac{99}{1101}= \frac{a_1}{m+r}$. This example can be generalized with $A = \{ k - 1, \overbrace{k+1,...,k+1}^{k^{1.5}}\}, m = k, T = k+1$, and yield $\rho_{a_1}(A,m,T) = \frac{\binom{k}{2}}{\binom{k}{2} + k^{1.5}} \frac{1}{3} \approx \frac{1}{3}$, $\frac{a_1}{m+r} = \frac{k-1}{k^{1.5} + k + 1} \approx \frac{1}{\sqrt{k}}$ for large enough values of $k$. Thus, not only is $\rho_{a_1}(A,m,T)$ not upper bounded by $\frac{a_1}{m+r}$, but it can supersede it by an arbitrary multiple. 
\end{example}

\section{A Power Index False-name Game}
\label{FalseNameGameSection}
In this section, we consider a framework general to all power indices under the possibility of votes split by big players in the voting game. Our discussion focuses on the Shapley-Shubik power index, where we give a conjecture with empirical results. 
%
We begin with a notation. Given a natural number $a$, define the integer partitions of $a$ as
$$Partitions(a) = \bigcup_{i=1}^{a} \left\{ \{b_1,...,b_i\} \bigg| \sum\limits_{j=1}^i b_j = a, \forall_{j=1}^i b_j \in \mathbb{N} \right\}$$ In words, the partitions of $a$ are all the different multi-sets of natural numbers such that their sum is $a$. Note that we allow several big players to split into multiple identities each, which is stronger than many other incentive analyses of false-name attack where only one strategic player is considered.

\begin{definition}[The false-name weighted voting game for a power index $\alpha \in \{\phi, \beta, \rho\} $]
Let $\{A,m,T\}$ be a WVG. We define a non-cooperative game with $|A|$ strategic players (which are the big players in the WVG). The strategy space of each strategic player is $Partitions(a_i)$. 
Given strategies $s_i = \{b_i^1,...,b_i^{c_i}\}$ for $1 \leq i \leq r$, let $B=\{b_1^1,...,b_1^{c_1},...,b_r^1,...,b_r^{c_r}\}$.
The payoff for player $i$ is
$$u_i^{\alpha}(s_1,...,s_r) = \sum\limits_{j=1}^{c_i} \alpha_{b_i^j}(B,m,T).$$
\noindent
Let $c=\sum\limits_{i=1}^r c_i$ stand for the total number of elements in $B$.
\end{definition}

We wish to understand how the option to split (submit false-name bids) changes the power of the strategic players. The previous section sheds light on this question:

\begin{theorem}
\label{ShapleyGameTheorem}
When $\alpha \in \{\phi, \rho\}$ (the Shapley-Shubik and Deegan-Packel indices), then for any tuple of strategies $s_1,...,s_r$,
$$
\sum\limits_{i=1}^r u_i^{\alpha}(s_1,...,s_r) \leq
\bar{R}_{\alpha} P(A,m,T).$$
In particular, this happens in any mixed or pure Nash or correlated equilibrium of the game.

\end{theorem}

Recall that $\bar{R}_{\phi} = 2, \bar{R}_{\rho} \leq 3$ (by Theorems~\ref{fairnessTheorem},\ref{thm:DPpropRatio} respectively). Note that by Theorem~\ref{banzhafUnbounded}, this theorem can not hold with any finite bound for the Banzhaf index. 

\begin{proof}
Given $B_{s_1,...,s_r}$ as the set of big players in the WVG after the application of the splitting strategies, we have 
$$\displaystyle \sum\limits_{i=1}^r \sum\limits_{j=1}^{c_i} \alpha_{b_i^j}(B_{s_1,...,s_r},m,T) \leq \bar{R}_{\alpha} \frac{\sum\limits_{i=1}^r \sum\limits_{j=1}^{c_i} b_i^j}{m +\sum\limits_{i=1}^r\sum\limits_{j=1}^{c_i} b_i^j} = \bar{R}_{\alpha} \frac{\sum\limits_{j=1}^r a_j}{m + \sum\limits_{j=1}^r a_j}.$$
\end{proof}

Thus, strategically splitting vote weights using false-name manipulations cannot increase the overall Shapley-Shubik power of the big players to be more than double their proportional weight. On the other hand, we conjecture that such strategic manipulations cannot harm their overall Shapley-Shubik power too much: 

\begin{conjecture}
\label{ShapleyConjecture}
For the Shapley-Shubik index $\phi$ with any weighted voting game $\{A,m,T\}$ and a choice of strategies resulting in a corresponding weighted voting game $\{B,m,T\}$, it holds that

$\displaystyle \sum\limits_{i=1}^r \phi_{a_i}(A,m,T) \leq 2\sum\limits_{i=1}^r \sum\limits_{j=1}^{c_i}\phi_{b_i^j}(B,m,T).$
\end{conjecture}

\begin{remark}
~~
\begin{itemize}

\item Section~\ref{experimentalResults} gives empirical results supporting the conjecture, obtained from an exhaustive search over small WVGs. 

\item 
Theorem~\ref{fairnessTheorem} is a special case of the conjecture,
where each player $i$'s strategy choice is $\overbrace{\{1,...,1\}}^{a_i}$.

\item Theorem 6 of \citeA{aziz2011false} states that a single player that splits her votes to exactly two identities cannot decrease its Shapley-Shubik index by more than a factor of $\frac{n+1}{2}$. Our conjecture gives a much stronger bound for the aggregate power of big players: The worst decrease of aggregate power, caused by any combination of splits, is by a constant factor of 2.

\end{itemize}
\end{remark}

\noindent
Combining Theorem~\ref{ShapleyGameTheorem} and Conjecture~\ref{ShapleyConjecture} yields:

\begin{corollary}
\label{FalseNameCorollary}
For all strategies $s_1,...,s_r$ in the Shapley false-name WVG $\{A,m,T\}$, if Conjecture~\ref{ShapleyConjecture} holds,
$$\frac{1}{2}\sum\limits_{i=1}^r \phi_{a_i}(A,m,T) \leq \sum\limits_{i=1}^r u_i^{\phi}(s_1,...,s_r) \leq 2P(A,m,T).$$
\end{corollary}


Our conclusions regarding the Shapley-Shubik power index are therefore the following:
false-name attacks can unboundedly increase the aggregate Shapley-Shubik power index of the big players, e.g., by splitting to singletons (Example~\ref{reverseRatio}). However, no attack can increase the power to more than twice the power resulting from the simple attack of splitting to singletons (Theorem~\ref{ShapleyGameTheorem}). False-name attacks can also decrease the Shapley-Shubik index (Example~\ref{lowerBoundShapley}). However, we believe, as expressed in Conjecture~\ref{ShapleyConjecture}, that no false-name attack can decrease the aggregate Shapley-Shubik power to be less than one-half of the original power.

Example~\ref{banzhafUnbounded} shows that an adaptation of Conjecture~\ref{ShapleyConjecture} does not hold for the Banzhaf index. For the Deegan-Packel index, even though  Theorem~\ref{thm:DPpropRatio} holds (similar to the Shapley-Shubik index case, and unlike the Banzhaf index case), an analogue of the conjecture fails to hold, as demonstrated in the following example:
\begin{example}
\label{ex:DPconjectureFails}
Consider $A = \{16\}, B = \{8,8\}, m = 16, T = 17$. For the WVG $(A,m,T)$ there are $16$ all-pivotal coalitions, all consisting of $a_1$ and some single small player, and so $\rho_{a_1}(A,m,T) = \frac{1}{2}$. For the WVG $(B,m,T)$, there are $16$ all-pivotal coalitions that include both $b_1, b_2$. There are $\binom{16}{9} = 11440$ all-pivotal coalitions with only $b_1$ (out of the big players), and the same amount with only $b_2$. There are no coalitions with no big players. Overall we have 
$$\rho_{b_1}(B,m,T) + \rho_{b_2}(B,m,T) = 2\rho_{b_1}(B,m,T) = 2 \frac{16 \cdot \frac{1}{3} + 11440 \cdot \frac{1}{10}}{22896} \approx 0.1004,$$
and so the ratio between the aggregate power of the big players before and after the split is approximately 5, larger than $\bar{R}_{\rho} \leq 3$. The example can further be generalized with $A = \{2k\}, B = \{k,k\}, m = 2k, T = 2k+1$. Then $\rho_{a_1}(A,m,T) = \frac{1}{2}, \rho_{b_1}(B,m,T) + \rho_{b_2}(B,m,T) = 2\frac{2k \cdot \frac{1}{3} + \binom{2k}{k+1} \cdot \frac{1}{k+2}}{2\binom{2k}{k+1} + 2k} \leq \frac{k + \binom{2k}{k+1} \cdot \frac{1}{k+2}}{\binom{2k}{k+1} + k} \stackrel{k\leq \binom{2k}{k+1} \cdot \frac{1}{k+2}}{\leq} \frac{2}{k+2}$ for large enough $k$. 
\end{example}

\subsection{The Worst-case Effects of False-name Manipulation for a Single Player}
\label{subsection:individualResults}
While the total utility of the big players is conjectured to not lose much by 
splits, a single player may multiplicatively lose arbitrarily much in B compared to A. This is evident by Example \ref{singlePlayerFairnessExample}, but we give two additional examples that do not require all players to fully split. In the first example, the player that chooses not to split loses by this choice. In the second example, the player that chooses to split loses by this choice. 
\begin{example}
\label{ex:loseByNotSplit}
Consider $A = \{k,k,k\}, B = \{k,\overbrace{1,...,1}^k,\overbrace{1,...,1}^k\}, m = k, T = 4k$. Then
\[
\begin{split}
    & \phi_{a_1}(A,m,T) = \frac{1}{k+3} \qquad  \phi_{b_1^1}(B,m,T) = \frac{1}{3k+1}.
\end{split}
\]
Thus, with $k\geq 6$, the ratio is higher than 2. The example can be generalized to exceed any bound $r$, with 
$A = \overbrace{\{k,...,k\}}^{r+1}, B = \{k,\overbrace{\overbrace{1,...,1}^k,...,\overbrace{1,...,1}^k}^{r}\}, m = k, T = (r+2)k.$
\end{example}

\begin{example}
\label{ex:loseBySplit}
Consider $A = \{k,k\}, B = \{\overbrace{1,...,1}^k,k\}, m = 0, T = k + 1$. Then
\[
\begin{split}
    & \phi_{a_1}(A,m,T) = \frac{1}{2} \qquad \sum\limits_{j=1}^{a_1}\phi_{b_1^j}(B,m,T) = \frac{1}{k+1}.
\end{split}
\]
\end{example}

The basic upper bound on the power of a single big player that Theorem~\ref{lemmaBasicUpperbound} yields, continues to hold under the possibility of splits, and it decreases as the number of splits increases:
\begin{corollary}
\label{FalseNameBasicUpperBound}
For a player $i$, and any strategy choice $s_1,...,s_r$ of the players, it holds that:
\[
\begin{split}
 & \boldsymbol{u_i^{\phi}(s_1,...,s_r)} = \sum\limits_{j=1}^{c_i} \phi_{b_i^j}(B,m,T) \stackrel{THM~\ref{lemmaBasicUpperbound}}{\leq} \\
 & \sum\limits_{j=1}^{c_i} \frac{b_i^j}{m + \sum\limits_{k=1}^r c_k}  \leq \sum\limits_{j=1}^{c_i} \frac{b_i^j}{m+r+(c_i-1)} = \\
 &\boldsymbol{\frac{a_i}{m+r+(c_i-1)}}.
\end{split}
\]
\end{corollary}
Note that we derive this corollary only for the Shapley-Shubik index, as examples~\ref{banzhafUnbounded} and~\ref{ex:DPnoBasicLowerBound} serve as counter-examples for the Banzhaf and Deegan-Packel indices respectively. 

\begin{corollary}
When $\alpha \in \{\phi, \rho\}$, then in any setting with a single big player $a_1$, $m$ small players, and any threshold $T$, the big player has a strategy that guarantees a fraction of at least $1/\bar{R}_{\alpha}$ of the power of its best possible strategy. In particular, this is the ``full split'' strategy. 
\end{corollary}
\begin{proof}
An immediate application of Theorem~\ref{ShapleyGameTheorem} yields 
that for any strategy $s_1$, 
$$u_1^{\phi}(s_1) \leq \bar{R}_{\alpha} P(A,m,T). $$
Thus,
the proportional value for the big player is at least $1/\bar{R}_{\alpha}$ as good as the best possible strategy. By the symmetry property of the Shapley-Shubik index, $s_1 = \overbrace{\{1,...,1\}}^{a_1}$ (``full split'') guarantees the proportional value. 
\end{proof}

It is important to note that not every strategy guarantees this ratio, as demonstrated by Example~\ref{lowerBoundShapley}. In the example, a player choosing not to split (or split into large pieces) will achieve much less than her optimal strategy. 



\subsection{Experimental Results}
\label{experimentalResults}
To support Conjecture~\ref{ShapleyConjecture}, we ran an exhaustive validation over all WVGs with $m<25, \sum\limits_{j=1}^r a_j < 25$.
The python code can be found at
\href{https://github.com/yotam-gafni/shapley\_fairness\_result}{github.com/yotam-gafni/shapley\_fairness\_result}. 
%
A total of $5,833,920$ WVGs were checked against all valid sub-partitions of them, resulting in a total of $1,246,727,916$ valid pairs being compared. The maximal ratio attained was $1.958333'$. The minimal ratio attained was $0.08$.
The exhaustive search consists of two parts. First, using dynamic programming over a dual recursion to that of Lemma~\ref{recursionLemma} (see Appendix~\ref{recursionAppendix}), we built a full recursion table of all Shapley-Shubik indices for the WVGs in the range. Then, for each WVG we considered all valid partition strategy sets $B$.

While the maximal ratio over all instances of the experimental analysis was close to $2$, in most instances the ratio was much closer to 1. Figure~\ref{RatioHistograms} shows a histogram of the number of cases (on the y-axis) for different possible ratios between 0 and 2 (on the x-axis). As can be seen from the figure, the ratio is concentrated around 1.

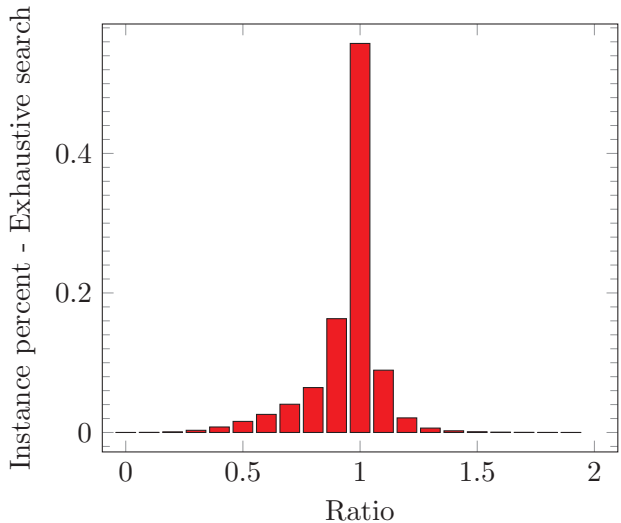
\begin{figure}
\centering
\scalebox{1}{
\begin{tikzpicture}
\begin{axis}[
	ylabel=Instance percent - Exhaustive search,
	xlabel=Ratio,
	enlargelimits=0.05,
	minor y tick num=9,
xmin=0, xmax=2,
axis on top]

\addplot [
  ybar,
  bar width=0.102874in, 
  bar shift=0in,
  fill=red,
  draw=black]
	coordinates {(1.0,0.5576346226613249) (0.7,0.04047445344923198) (1.4,0.002403652762997889) (0.6,0.02592546183108007) (1.3,0.0063072382506914205) (1.2,0.020914454280977213) (0.8,0.06438205880359865) (1.1,0.08931312965001419) (0.9,0.16313806917274482) (0.5,0.015885208589489865) (1.6,0.0004254103828056097) (0.4,0.00787652772828422) (0.3,0.0031184695153645696) (1.5,0.0010209974314876896) (1.7,0.00015063671679266385) (0.2,0.0008634289696934965) (1.8,3.332483332313544e-05) (0.1,0.000130638768820189) (1.9,1.3354958837706816e-06) (0.0,8.807053936217467e-07)};
\end{axis}
\end{tikzpicture}
}

\caption{Ratio of big players' power before and after splits}
\label{RatioHistograms}
\end{figure}

\section{Discussion and Future Directions}

Many questions remain open. We find Conjecture~\ref{ShapleyConjecture} hard to prove even in limited settings. For example, consider the WVG $\{B,m,T\}$ where $m<T<\sum_{i=1}^{|B|}b_i$, i.e., the overall weights of the small players are less than the threshold, which itself is less than the overall weights of the big players. In this case, it is possible to show that if we take $A = \{\sum_{i=1}^{|B|}b_i\}$, i.e., a single big player, then that player has a Shapley-Shubik index of 1 in the WVG $\{A,m,T\}$. The conjecture’s inequality in that case then states that the sum of Shapley-Shubik indices of the big players in $B$ is larger or equal to $\frac{1}{2}$. This reads as a very clean combinatorial problem: If we draw a permutation at random over a multi-set of integers $m\times \{1\} \cup B$, with $m<T<\sum_{i=1}^{|B|}b_i$, then the probability that the pivotal player (crossing the threshold $T$) is ``big'' is higher than the probability that it is ``small''. This can be even simplified further if we assume all big players are of identical size $k$.

The model itself could be generalized so that the threshold value $T$ and big players' values $A$ are not a multiple of the small players value, or into some other idea of looser distinctions between big and small players. For example, the Israeli Parliament currently has a minimal threshold of $3$ seats for a party to enter parliament. Thus, prior to elections, it does not make sense for a bloc to split into parties with less than at least $3$ seats each (as indicated by polls). Therefore, maintaining our notion that ``small'' blocs are these that can not split, any player $a$ with size $3 \leq a < 6$ is considered small. Any other player is considered big. Since not all small players now must have the same size, we no longer parameterize them by the number $m$, but by the multi-set $M$ of small players. As for $T$, consider that the threshold to pass a basic law in the Israeli Parliament is an absolute majority of $61$. This is not a multiple of $s=3$, unlike the case of $s=1$ we consider throughout the paper. As a concrete example, consider a WVG with weights $A = \{8,7\}, M = \{5,3\}, T = 10$. Then, $P(A,M,T) = \frac{8 + 7}{8 + 7 + 5 + 3} = \frac{15}{23}, \sum_{j=1}^{2}\phi_{a_j}(A,M,T) = \frac{2}{3}$, and so $\frac{\sum_{j=1}^{2}\phi_{a_j}(A,M,T)}{P(A,M,T)} = \frac{46}{45} \leq 2$. As we can see, the claim of Theorem~\ref{fairnessTheorem} is still meaningful (with appropriate adjustments) in these settings. In order to test it, we ran an exhaustive validation over all WVGs with $\sum_{j=1}^r a_j < \sum_{j=1}^{|M|} m_j < 25$. 
We find that over this set,\footnote{We remind the reader that if $\sum_{j=1}^r a_j \geq \sum_{j=1}^{|M|} m_j$, then $P(A,M,T) \geq \frac{1}{2}$, and so the ratio must be at most $2$. For this reason we forgo exact calculation for these WVGs. } indeed $\bar{R}_{\phi} \leq 2$. A total of $90,141$ WVGs were checked. The maximal ratio attained was $\approx 1.8585$, for the instance $A = \{22\}, M = \overbrace{\{3,...,3\}}^{8}, T = 22$. The minimal ratio attained was $\approx 0.22705$, for the instance $A = \{23\}, M = \overbrace{\{3,...,3\}}^{8}, T = 47$. We derived similar results for every $2 \leq s \leq 11$, which is the relevant range of interest for the WVGs examined. The results are summarized in Table~\ref{tab:robust-check} (ratios are given approximately). 

\begin{table}[]
    \centering
    \begin{tabular}{c|c|c|c}
         \textbf{s} & \textbf{\# Instances} & \textbf{Max ratio} & \textbf{Min ratio}\\
         \hline
         2 & $161,737$ & $1.88628$ & $0.15719$ \\
         \hline
         3 & $90,141$ & $1.8585$ & $0.22705$ \\
         \hline
         4 & $46,262$ & $1.83673$ & $0.291925$ \\
        \hline
         5 & $23,386$ & $ 1.8666$ & $0.40869$ \\
        \hline
         6 & $12,691$ & $1.88888$ & $0.40869$ \\
        \hline
         7 & $8,075$ & $1.8095$ & $0.510869$ \\
        \hline
         8 & $5,310$ & $1.8088$ & $0.510869$ \\
        \hline
         9 & $2,642$ & $1.5555$ & $0.681159$ \\
        \hline
         10 & $1,021$ & $1.46666$ & $0.681159$ \\
        \hline
         11 & $231$ & $1.3939$ & $0.681159$ \\
    \end{tabular}
    \caption{Aggregate big players' power vs. proportion ratios with varying ``small'' definition}
    \label{tab:robust-check}
\end{table}

Other interesting extensions for the work can be deriving tight bounds for the Deegan-Packel index, and exploring similar results for other power indices in common use, such as these of \citeA{johnston1978measurement}, 
\citeA{hollerPackel}, and  
\citeA{coleman1971Control}. 
Generalizing our results to a larger class of cooperative games is also interesting. 

\section*{Acknowledgements}
A preliminary version of this paper has appeared in the proceedings of the 30th International Joint Conference on Artificial Intelligence (IJCAI-21), and was presented at the 8th International Workshop on Computational Social Choice (COMSOC-2021). We thank all anonymous reviewers for their helpful comments. 

Yotam Gafni and Moshe Tennenholtz were supported by the European Research Council (ERC) under the European Union’s Horizon 2020 research and innovation programme (Grant No. 740435).  

\noindent Ron Lavi was partially supported by the ISF-NSFC joint research program (grant No. 2560/17).

\appendix

\section{Proof of Step (2) in the \ref{recursionStep1} Equation, Used in the Proof of Lemma~\ref{shapleyThmInductionProof}}
\label{app:step2}

\begin{lemma}
\label{shapleyInductionInequalities}
\[
\begin{split}
& \frac{m}{m+r} \cdot \frac{m-1 - \mathlarger{\sum}\limits_{j=2}^T r_j \min(j, T-1)}{m-1+\displaystyle \sum\limits_{j=2}^T r_j\min(j, T-1)} + \\
& \mathlarger{\mathlarger{\mathlarger{\sum\limits_{i=2}^{T-1}}}} \frac{r_i}{m+r} \cdot \frac{m  -\left(\mathlarger{\sum}\limits_{2\leq j\neq i \leq T} r_j \min(j, T-i)\right) - (r_i - 1)\min(i,T-i)}{m+\left(\mathlarger{\sum}\limits_{2\leq j\neq i\leq T} r_j \min(j, T-i)\right)+ (r_i - 1)\min(i,T-i)} \geq  \frac{m  -  \displaystyle \sum\limits_{j=1}^r a_j }{m + \displaystyle \sum\limits_{j=1}^r a_j}
\end{split}
\]
\end{lemma}
\begin{proof}

We define (within the following equations), $G^{-}(A,m,T,i),G^{+}(A,m,T,i)$ for shorter notations. Observe that for any weight $2\leq i \leq T-1$,
\[
\begin{split}
& G^{-}(A,m,T,i) \stackrel{def.}{=} m  -\left(\sum\limits_{2\leq j\neq i \leq T} r_j \min(j, T-i)\right) - (r_i - 1)\min(i,T-i) \geq \\
& m  - \left(\sum\limits_{2\leq j\neq i \leq T-1} r_j\cdot j\right) - r_T \cdot (T-1) - (r_i - 1)i = m  - \sum\limits_{j=2}^T r_j \cdot j + r_T + i= m - \sum\limits_{j=1}^r a_j + r_T + i,
\end{split}
\]
and similarly, 
\[
\begin{gathered}
 G^{+}(A,m,T,i) \stackrel{def.}{=} m  +\left(\sum\limits_{2\leq j\neq i \leq T} r_j \min(j, T-i)\right) + (r_i - 1)\min(i,T-i) \leq m + \sum\limits_{j=1}^r a_j - r_T - i,  \\
  \sum\limits_{j=2}^T r_j \min(j, T-1) = \sum\limits_{j=1}^r a_j - r_T.
\end{gathered}
\]

By making the relevant substitutions to the LHS of the lemma's claim, we thus have, 

\begin{equation}
\begin{split}
    & \frac{m}{m+r} \cdot \frac{m-1-\displaystyle \sum\limits_{j=2}^T r_j \min(j, T-1)}{m-1+\displaystyle \sum\limits_{j=2}^T r_j \min(j, T-1)} + \mathlarger{\mathlarger{\mathlarger{\sum\limits_{i=2}^{T-1}}}}\frac{r_i}{m+r} \cdot \frac{G^{-}(A,m,T,i)}{G^{+}(A,m,T,i)}
    \geq \\
& \frac{m}{m+r} \cdot \frac{m - \displaystyle \sum\limits_{j=1}^r a_j + r_T - 1}{m  + \displaystyle \sum\limits_{j=1}^r a_i - r_T - 1} + \mathlarger{\mathlarger{\mathlarger{\sum\limits_{i=2}^{T-1}}}} \frac{r_i}{m+r} \cdot \frac{m  - \displaystyle\sum\limits_{j=1}^r a_i + r_T + i}{m  + \displaystyle \sum\limits_{j=1}^r a_i - r_T -i}.
    \end{split}
\end{equation}

For the case where $r_T = 1$ we further have
\begin{equation}
    \label{caseEq1}
\begin{split}
    &\frac{m}{m+r} \cdot \frac{m - \sum\limits_{j=1}^r a_j + r_T - 1}{m + \sum\limits_{j=1}^r a_j - r_T - 1} + \sum\limits_{i=2}^{T-1} \frac{r_i}{m+r} \cdot \frac{m  - \sum\limits_{j=1}^r a_j + r_T + i}{m + \sum\limits_{j=1}^r a_j - r_T -i} = \\
    & \frac{m}{m+r} \cdot \frac{m  -  \sum\limits_{j=1}^r a_j }{m + \sum\limits_{j=1}^r a_j - 2}  + \sum\limits_{i=2}^{T-1} \frac{r_i}{m+r} \cdot \frac{m  - \sum\limits_{j=1}^r a_j +i + 1}{m + \sum\limits_{j=1}^r a_j - i - 1} \stackrel{(1)}{\geq} \\
    &  \frac{m+1}{m+r}\cdot \frac{m  -  \sum\limits_{j=1}^r a_j }{m + \sum\limits_{j=1}^r a_j} + \sum\limits_{i=2}^{T-1} \frac{r_i}{m+r} \cdot \frac{m  -  \sum\limits_{j=1}^r a_j }{m + \sum\limits_{j=1}^r a_j} = \\
    & \frac{m + \sum\limits_{i=2}^{T-1}r_i + r_T}{m+r} \cdot \frac{m  -  \sum\limits_{j=1}^r a_j }{m + \sum\limits_{j=1}^r a_j} = \frac{m  -  \sum\limits_{j=1}^r a_j }{m + \sum\limits_{j=1}^r a_j},
 \end{split}
 \end{equation}
 
 where for $(1)$ note that by arithmetics, $\displaystyle \frac{m}{m + \sum\limits_{j=1}^r a_j - 2} \geq \frac{m+1}{m + \sum\limits_{j=1}^r a_j}$. 

We next treat the case where $r_T \neq 1$. 
\begin{equation}
\label{inductionStep2}
\begin{split}
& \frac{m}{m+r} \cdot \frac{m - \sum\limits_{j=1}^r a_j + r_T - 1}{m + \sum\limits_{j=1}^r a_i - r_T - 1} + \sum\limits_{i=2}^{T-1} \frac{r_i}{m+r} \cdot \frac{m  - \sum\limits_{j=1}^r a_j + r_T + i}{m \sum\limits_{j=1}^r a_j - r_T -i} \stackrel{(1)}{\geq} \\
& \frac{m  -  \sum\limits_{j=1}^r a_j}{m+ \sum\limits_{j=1}^r a_j - r_T} + \frac{(\sum\limits_{j=1}^r a_j + r - r_T)r_T}{(m+r)(m + \sum\limits_{j=1}^r a_j - r_T)} - \frac{2m r_T T}{(m+r)(m+ \sum\limits_{j=1}^r a_j - r_T - 1)(m+ \sum\limits_{j=1}^r a_j - r_T)} \stackrel{(2)}{\geq} \\
& \frac{m - \sum\limits_{j=1}^r a_j}{m + \sum\limits_{j=1}^r a_j}
\end{split}
\end{equation}
We develop (1) in Appendix~\ref{detailedInductionStep2Ext}. We prove (2) by examining two cases, based on the value of $r_T$. 

\begin{enumerate}
\item If $r_T = 0$, it is immediate by assignment. 

\item If $r_T \geq 2$, we have
\[
\begin{split}
& \overbrace{\sum\limits_{j=1}^r a_j +r - r_T}^{Exp1} \geq r_T \cdot T + r - r_T \geq \overbrace{2T}^{Exp2}, \qquad \overbrace{m + \sum\limits_{j=1}^r a_j - r_T - 1}^{Exp3}\geq m + r_T T - r_T - 1 \geq \overbrace{m}^{Exp4}
\end{split}
\]
which allows us to replace $Exp1$ with $Exp2$, cancel $Exp3$ and $Exp4$ with each other, and arrive at
\[
\begin{split}
   & \frac{m  -  \sum\limits_{j=1}^r a_j}{m+ \sum\limits_{j=1}^r a_j - r_T} + \frac{\overbrace{(\sum\limits_{j=1}^r a_j + r - r_T)}^{Exp1}r_T}{(m+r)(m + \sum\limits_{j=1}^r a_j - r_T)} - \frac{2\overbrace{m}^{Exp4}r_T T}{(m+r)\underbrace{(m+ \sum\limits_{j=1}^r a_j - r_T - 1)}_{Exp3}(m+ \sum\limits_{j=1}^r a_j - r_T)} \geq \\
   & \frac{m  -  \sum\limits_{j=1}^r a_j}{m+ \sum\limits_{j=1}^r a_j - r_T} + \frac{\overbrace{2T}^{Exp2} r_T}{(m+r)(m + \sum\limits_{j=1}^r a_j - r_T)} - \frac{2r_T T}{(m+r)(m+ \sum\limits_{j=1}^r a_j - r_T)} = \\
   & \frac{m  -  \sum\limits_{j=1}^r a_j}{m+ \sum\limits_{j=1}^r a_j - r_T} \geq \frac{m  -  \sum\limits_{j=1}^r a_j}{m+ \sum\limits_{j=1}^r a_j}.
\end{split}
\]

\end{enumerate}

\end{proof}

\section{Expansion of the Steps to Derive (1) in Eq.~\ref{inductionStep2}}
\label{detailedInductionStep2Ext}

A note on notation - before each equality/inequality, we mark the expressions that are to be manipulated in the subsequent step. We note that these expressions are not related to the next ones after the equality/inequality sign. For example, $Exp1$ and $Exp2$ are central to transformation $(1)$, but they are unrelated to $Exp3$ and $Exp4$, which are marked because they are central to transformation $(2)$. After giving the full chain of inequalities, we explain each transformation.

\begin{align*}
& \frac{m}{m+r} \cdot \overbrace{\frac{m  -  \sum\limits_{j=1}^r a_j + r_T - 1}{m + \sum\limits_{j=1}^r a_j - r_T - 1 }}^{Exp1} + \sum\limits_{i=2}^{T-1} \frac{r_i}{m+r} \cdot \overbrace{\frac{m - \sum\limits_{j=1}^r a_j + r_T + i}{m + \sum\limits_{j=1}^r a_j - r_T - i}}^{Exp2} \stackrel{(1)}{=} \\
& \frac{m}{m+r} \cdot \left(\frac{m  -  \sum\limits_{j=1}^r a_j + r_T}{m+ \sum\limits_{j=1}^r a_j - r_T} +\frac{\overbrace{2r_T}^{Exp3} - 2\sum\limits_{j=1}^r a_j}{(m+ \sum\limits_{j=1}^r a_j - r_T - 1)(m+ \sum\limits_{j=1}^r a_j - r_T)}\right) + \\
& \frac{1}{m+r}\sum\limits_{i=2}^{T-1} r_i \left(\frac{m - \sum\limits_{j=1}^r a_j + r_T}{m + \sum\limits_{j=1}^r a_j  - r_T} + \frac{2mi}{\underbrace{(m+ \sum\limits_{j=1}^r a_j - r_T - i)}_{Exp4}(m+ \sum\limits_{j=1}^r a_j - r_T)}\right) \stackrel{(2)}{\geq} \\
& \frac{m}{m+r} \cdot \left(\overbrace{\frac{m  -  \sum\limits_{j=1}^r a_j + r_T}{m+ \sum\limits_{j=1}^r a_j - r_T}}^{Exp5} -\frac{2\sum\limits_{j=1}^r a_j}{(m+ \sum\limits_{j=1}^r a_j - r_T - 1)(m+ \sum\limits_{j=1}^r a_j - r_T)}\right) + \\
& \frac{1}{m+r}\sum\limits_{i=2}^{T-1} r_i \left(\overbrace{\frac{m  - \sum\limits_{j=1}^r a_j + r_T}{m  + \sum\limits_{j=1}^r a_j - r_T}}^{Exp5} + \frac{2mi}{(m+ \sum\limits_{j=1}^r a_j - r_T - 1)(m+ \sum\limits_{j=1}^r a_j - r_T)}\right) \stackrel{(3)}{=} \\
& \overbrace{\frac{m + \sum\limits_{i=2}^{T-1}r_i}{m+r}}^{Exp6} \cdot \frac{m\!-\!\sum\limits_{j=1}^r a_j\!+\!r_T}{m\!+\!\sum\limits_{j=1}^r a_j\!-\!r_T} + \frac{\overbrace{2m\sum\limits_{i=2}^{T-1}r_i \cdot i - 2m\sum\limits_{j=1}^{r}a_j}^{Exp7}}{(m\!+\!r)(m\!+\! \sum\limits_{j=1}^r a_j - r_T - 1)(m+ \sum\limits_{j=1}^r a_j - r_T)} \stackrel{(4)}{=} \\
& \overbrace{\left(1\!-\!\frac{r_T}{m+r}\right) \cdot \frac{m\!-\!\sum\limits_{j=1}^r a_j\!+\!r_T}{m\!+\!\sum\limits_{j=1}^r a_j\!-\!r_T}}^{Exp8} - \frac{2mr_T T}{(m\!+\!r)(m\!+\! \sum\limits_{j=1}^r a_j - r_T - 1)(m+ \sum\limits_{j=1}^r a_j - r_T)} \stackrel{(5)}{=} \\
& \frac{m  -  \sum\limits_{j=1}^r a_j}{m+ \sum\limits_{j=1}^r a_j - r_T} + \frac{(\sum\limits_{j=1}^r a_j + r - r_T)r_T}{(m+r)(m + \sum\limits_{j=1}^r a_j - r_T)} - \\
& \frac{2mr_T T}{(m+r)(m+ \sum\limits_{j=1}^r a_j - r_T - 1)(m+ \sum\limits_{j=1}^r a_j - r_T)}. 
\end{align*}

Hence we arrive at the desired expression. We state in words the changes introduced in each step:
\begin{enumerate}
    \item We transform $Exp1$ by performing
   \begin{align*}
 &  \frac{m - \sum\limits_{j=1}^r a_j + r_T - 1}{m + \sum\limits_{j=1}^r a_j - r_T - 1} = \frac{m - \sum\limits_{j=1}^r a_j + r_T - 1}{m + \sum\limits_{j=1}^r a_j - r_T - 1} \cdot \frac{m + \sum\limits_{j=1}^r a_j - r_T -1 + 1}{m + \sum\limits_{j=1}^r a_j - r_T} = \\
 & \frac{m - \sum\limits_{j=1}^r a_j + r_T - 1}{m + \sum\limits_{j=1}^r a_j - r_T - 1} \cdot \left(\frac{m + \sum\limits_{j=1}^r a_j - r_T -1}{m + \sum\limits_{j=1}^r a_j - r_T}+\frac{1}{m + \sum\limits_{j=1}^r a_j - r_T}\right) = \\
 & \frac{m - \sum\limits_{j=1}^r a_j + r_T - 1}{m + \sum\limits_{j=1}^r a_j - r_T} + \frac{m - \sum\limits_{j=1}^r a_j + r_T - 1}{(m + \sum\limits_{j=1}^r a_j - r_T - 1)(m + \sum\limits_{j=1}^r a_j - r_T)} = \\
 & \frac{m - \sum\limits_{j=1}^r a_j + r_T}{m + \sum\limits_{j=1}^r a_j - r_T} + \frac{(m - \sum\limits_{j=1}^r a_j + r_T - 1) - (m + \sum\limits_{j=1}^r a_j - r_T - 1)}{(m + \sum\limits_{j=1}^r a_j - r_T - 1)(m + \sum\limits_{j=1}^r a_j - r_T)} = \\
 & \frac{m  -  \sum\limits_{j=1}^r a_j + r_T}{m+ \sum\limits_{j=1}^r a_j - r_T} +\frac{2r_T - 2\sum\limits_{j=1}^r a_j}{(m+ \sum\limits_{j=1}^r a_j - r_T - 1)(m+ \sum\limits_{j=1}^r a_j - r_T)}
 \end{align*}

    A similar transformation is applied to $Exp2$. 
    
    \item We omit $Exp3$. We substitute $Exp4$, which is $(m + \sum\limits_{j=1}^r a_j -r_T - i)$, with $(m + \sum\limits_{j=1}^r a_j -r_T - 1)$. 
    
    \item We gather the terms multiplying $Exp5$  (which is $\frac{m - \sum\limits_{j=1}^r a_j + r_T}{m + \sum\limits_{j=1}^r a_j -r_T}$) as the first summand, and all other terms as the second summand. 
    
    \item Since $r = \sum\limits_{j=2}^T r_j$, we can rewrite $Exp6$ as
    $\displaystyle \frac{m + \sum\limits_{i=2}^{T-1} r_i}{m + r} = \frac{m + r - r_T}{m + r} = 1 - \frac{r_T}{m+r}$. 
    
    Since $\sum\limits_{j=1}^r a_j = \sum\limits_{i=2}^T r_i \cdot i$, we can rewrite $Exp7$
    $$ 2m \sum\limits_{i=2}^{T-1} r_i \cdot i - 2m \sum\limits_{j=1}^r a_j =  2m \sum\limits_{i=2}^{T-1} r_i \cdot i - 2m\sum\limits_{j=2}^T r_j \cdot j = -2m r_T \cdot T.$$
    
    \item We rewrite $Exp8$ 
    \[
    \begin{split}
        & \left(1 - \frac{r_T}{m+r}\right) \cdot \frac{m - \sum\limits_{j=1}^r a_j + r_T}{m + \sum\limits_{j=1}^r a_j - r_T} = \\
        & \frac{m - \sum\limits_{j=1}^r a_j}{m + \sum\limits_{j=1}^r a_j - r_T} \left(1 - \frac{r_T}{m + r}\right) + \frac{r_T}{m + \sum\limits_{j=1}^r a_j - r_T} \cdot \frac{m + r - r_T}{m + r} = \\
        & \frac{m - \sum\limits_{j=1}^r a_j}{m + \sum\limits_{j=1}^r a_j - r_T} - \frac{r_T}{m+r} \cdot \frac{m - \sum\limits_{j=1}^r a_j}{m + \sum\limits_{j=1}^r a_j - r_T} + \frac{r_T}{m + \sum\limits_{j=1}^r a_j - r_T} \cdot \frac{m+r - r_T}{m+r} = \\
        & \frac{m - \sum\limits_{j=1}^r a_j}{m + \sum\limits_{j=1}^r a_j - r_T} + \frac{r_T( \sum\limits_{j=1}^r a_j + r - r_T)}{(m+r)(m + \sum\limits_{j=1}^r a_j - r_T)}. 
    \end{split}
    \]
\end{enumerate}

\section{Proof of Theorem~\ref{banzhafUnbounded}}
\label{appendix:banzhafUnbounded}

\begin{proof} 
Choose $k = (2n_k)^2$ for some natural number $n_k$, so $k = 4, 16, 36, ...$. This choice maintains the property that the weighted voting game parameters and all binomial coefficients' arguments throughout the proof are natural numbers. Consider $\displaystyle A = \{2k\}, m = k^{1.5}, T = \frac{2k + k^{1.5}}{2}$. Then
$P(A,m,T) = \frac{2k}{2k + k^{1.5}} \in o(1)$
and
$$ v(S\cup \{a_1\}) - v(S) = \begin{cases}
1 & \frac{k^{1.5}}{2} - k \leq |S| < \frac{k^{1.5}}{2} + k \\ 
0 & otherwise.
\end{cases}$$
The absolute Banzhaf index gives an equal probability to each subset in $P((A\setminus \{a_1\})\cup M) = P(M)$. There are $2^{k^{1.5}}$ subsets, and by the above, each one of those adds one to the index iff it has $i$ small players as specified above. Thus,
\begin{equation}
\label{BanzhafBinom}
\begin{split}
    & \beta'_{a_1}(A,m,T) = \hspace{-5ex} \mathlarger{\sum}\limits_{\substack{S \in P(M) \\ \frac{k^{1.5}}{2} - k \leq |S| < \frac{k^{1.5}}{2} + k}}\hspace{-5ex} \frac{1}{2^{k^{1.5}}} = 
 \frac{1}{2^{k^{1.5}}} \mathlarger{\sum\limits_{i=\frac{k^{1.5}}{2} - k}^{\frac{k^{1.5}}{2} + k - 1}}\binom{k^{1.5}}{i} \geq \\
 & Pr_{X\sim B(k^{1.5},\frac{1}{2})}\left[|X - \frac{k^{1.5}}{2}| < k \right] = \\
 & 1 - Pr_{X\sim B(k^{1.5},\frac{1}{2})}\left[|\frac{X}{k^{1.5}} - \frac{1}{2}| \geq \frac{1}{\sqrt{k}}\right],
     \end{split}
\end{equation}
where $B(n,p)$ is the binomial distribution with $n$ players and success probability $p$. 

We cite a Chernoff-type inequality (By Theorem 1 of \citeR{binomialChernoff}), for $\displaystyle X\sim B(n,p), c>0$:
\begin{equation}
Pr_{X\sim B(n,p)}\left[|\frac{X}{n} - p| \geq c\right] < 2e^{-2nc^2}
\end{equation}
If we choose $n = k^{1.5}, p = \frac{1}{2}, c = \frac{1}{\sqrt{k}}$, we get
\begin{equation}
    \label{OurChernoff}
    Pr_{X\sim B(k^{1.5},\frac{1}{2})}\left[|\frac{X}{k^{1.5}} - \frac{1}{2}| \geq \frac{1}{\sqrt{k}}\right] < 2e^{-2\sqrt{k}} \leq 0.05
\end{equation}

So Eq.~\ref{BanzhafBinom} and Eq.~\ref{OurChernoff} yield 
\begin{equation}
    \label{bigPlayerAbsBanzhafLowerBound}
    \beta'_{a_1}(A,m,T) \geq 0.95.
\end{equation} 
This concludes the argument for the absolute Banzhaf index. For the normalized Banzhaf index, we first need to bound the small player's absolute Banzhaf index. Notice that for a fixed small player to be pivotal, the sum of subset $S$ elements needs to be exactly $T -1 = \frac{k^{1.5}}{2} + k - 1$. This means that either $a_1 \in S$ and there are $\frac{k^{1.5}}{2} - k - 1$ small players (other than our one fixed small player) in the subset, or $a_1 \not \in S$ and there are $\frac{k^{1.5}}{2} + k - 1$ other small players in the subset. Thus,
\begin{equation}
\label{smallPlayerAbsBanzhafBound}
\begin{split}
    & \beta'_1(A,m,T) = Pr[a_1 \in S] \frac{1}{2^{m-1}} \binom{m-1}{\frac{k^{1.5}}{2} - k -1} + \\
    & Pr[a_1 \not \in S] \frac{1}{2^{m-1}} \binom{m-1}{\frac{k^{1.5}}{2} + k -1} = \\
    &\frac{1}{2^{k^{1.5}}} \left[\binom{k^{1.5} - 1}{\frac{k^{1.5}}{2} + k - 1} + \binom{k^{1.5} - 1}{\frac{k^{1.5}}{2} - k - 1}\right] = \\
& \frac{1}{2^{k^{1.5}}} \left[\binom{k^{1.5} - 1}{\frac{k^{1.5}}{2} + k - 1} + \binom{k^{1.5} - 1}{\frac{k^{1.5}}{2} +k}\right]  \stackrel{(1)}{=} \\
& \frac{1}{2^{k^{1.5}}} \binom{k^{1.5}}{\frac{k^{1.5}}{2}+k} \leq Pr_{X\sim B(k^{1.5},\frac{1}{2})}\left[|X - \frac{k^{1.5}}{2}| \geq k \right] = \\
& Pr_{X\sim B(k^{1.5},\frac{1}{2})}\left[|\frac{X}{k^{1.5}} - \frac{1}{2}| \geq \frac{1}{\sqrt{k}}\right] \stackrel{Eq.\ref{OurChernoff}}{<} 2e^{-2\sqrt{k}},
\end{split}
\end{equation}
where $(1)$ is by Pascal identity. So for the normalized Banzhaf index, we can write

$ \displaystyle \beta_{a_1}(A,m,T) = \frac{\beta'_{a_1}}{k^{1.5} \beta'_1(A,m,T) + \beta'_{a_1}} \stackrel{Eq.~\ref{eqBanzhafIndividualUpperBound} +\ref{bigPlayerAbsBanzhafLowerBound} + \ref{smallPlayerAbsBanzhafBound}}\geq  \frac{0.95}{2k^{1.5} e^{-2\sqrt{k}} + 1} \geq 0.7
$
\end{proof}

\section{Proof of Theorem~\ref{thm:DPpropRatio} for the High Threshold Regime}
\label{appendix:DPFirstPartProof}

\begin{lemma}
\label{DPFirstPart}
If for some $A,m,T$, $\displaystyle \frac{\displaystyle \sum\limits_{i=1}^r \rho_{a_i}(A,m,T)}{P(A,m,T)} > 3$, then $T < \frac{m}{2}$ and $\displaystyle \sum\limits_{j=1}^r a_j < m$. 
\end{lemma}

\begin{proof}
By the lemma's premise, for some $A,m,T$ values,
\begin{equation}
\label{lemmaPrior}
    \frac{\displaystyle \sum\limits_{i=1}^r \rho_{a_i}(A,m,T)}{P(A,m,T)} > 3.
\end{equation}
First, note that we must have
$
\sum\limits_{j=1}^r a_j < \frac{m}{2}.
$
Otherwise,
$$
3 P(A,m,T)\!=\!3 \frac{\sum\limits_{j=1}^r a_j}{m\!+\! \sum\limits_{j=1}^r a_j}\!\geq\!\mathlarger{\frac{3\cdot  \frac{m}{2}}{m\!+\!\frac{m}{2}}}\!\stackrel{eff.}{\geq}\!\sum\limits_{i=1}^r \rho_{a_i}(A,m,T),
$$
\noindent
and dividing both sides by $P(A,m,T)$ contradicts Eq.~\ref{lemmaPrior}. 

Recall that $AP$ is the set of all-pivotal player subsets. Let $l = \min_{S\in AP} \left|S\setminus A\right|$, the least amount of small player participants in an all-pivotal subset $S$. We use the notation $x^+ = \max(x,0)$. It can be directly checked that it must hold that $
    l = (T - \sum\limits_{j=1}^r a_j)^+
$. 
Now,
\begin{equation}
\begin{split}
\label{upperBoundDPEq}
    & \sum\limits_{i=1}^r \rho_{a_i}(A,m,T) = \mathlarger{\mathlarger{\sum\limits_{i=1}^r}}\mathbb{E}_{S\sim UNI(AP)} \left[\frac{\mathbbm{1}_{a_i \in S}}{|S|}\right] = \\
    & \mathbb{E}_{S\sim UNI(AP)} \left[ \sum\limits_{i=1}^r \frac{\mathbbm{1}_{a_i \in S}}{|S|}\right] \stackrel{(1)}{\leq} \mathbb{E}_{S\sim UNI(AP)} \left[\frac{|S \cap A|}{l + |S \cap A|}\right] \leq \\
    & \mathbb{E}_{S\sim UNI(AP)} \left[\frac{r}{l + r}\right] \leq \frac{r}{l+r} = \frac{r}{(T - \displaystyle \sum\limits_{j=1}^r a_j)^+ + r},\\
    \end{split}
    \end{equation}
where (1) is since for the numerator, we have $\sum_{i=1}^r \mathbbm{1}_{{a_i}\in S} = |S \cap A|$. Also, for any all-pivotal set $S$, $l \leq |S\setminus A|$ by its definition, and so for the denominator we have $l + |S\cap A| \leq |S\setminus A| + |S\cap A| = |S|$. 
%
We thus have
\begin{equation}
\label{maxAPsetCondition}
\frac{r}{(T - \displaystyle \sum\limits_{j=1}^r a_j)^+ + r} \stackrel{Eq.~\ref{lemmaPrior}+\ref{upperBoundDPEq}}{>} 3 \frac{\displaystyle \sum\limits_{j=1}^r a_j}{m + \displaystyle \sum\limits_{j=1}^r a_j}. \end{equation}

Finally, we consider two cases. 
\begin{enumerate}

    \item If $\displaystyle T - \sum\limits_{j=1}^r a_j < 0$, we have $\displaystyle T < \sum\limits_{j=1}^r a_j < \frac{m}{2}$.

    \item If $\displaystyle T - \sum\limits_{j=1}^r a_j \geq 0$, we can make the substitution $(T - \sum_{j=1}^r a_j)^+ = T - \sum_{j=1}^r a_j$. We can also write $\sum_{j=1}^r a_j = \frac{m}{\theta}$ for some $\theta > 2$. We can then rearrange Eq.~\ref{maxAPsetCondition} and have
    \[
\begin{split}
& T < \frac{r(m + \sum\limits_{j=1}^r a_j)}{3\sum\limits_{j=1}^r a_j} + \sum\limits_{j=1}^r a_j - r \leq \frac{r(\theta - 2)}{3} + \frac{m}{\theta} \stackrel{(1)}{\leq} \\
& \frac{m(\theta - 2)}{6\theta} + \frac{m}{\theta} < \frac{m}{2}
\end{split}
\]
where (1) is since $ \displaystyle r \leq \frac{\sum_{j=1}^r a_j}{2} = \frac{m}{2\theta}$ (the big players are of size at least 2). 
\end{enumerate}
\end{proof}

\section{Proof of Theorem~\ref{thm:DPpropRatio} for the Low Threshold Regime}
\label{appendix:DPsecondPartProof}

\begin{restatable}{lemma}{DPsecondPart}
\label{DPsecondPart}
For any $A,m,T$ with $T < \frac{m}{2}$ and $\sum\limits_{j=1}^r a_j < m$, 
$$\displaystyle \frac{\sum\limits_{i=1}^r \rho_{a_i}(A,m,T)}{P(A,m,T)} \leq 2.$$ 
\end{restatable}

\begin{proof}
Let $a = \max A$ and $r_2,...,r_a$ the number of players with weights $2,...,a$ respectively. For some all-pivotal set $S\in AP$, let $i_j(S) = |\{i\}_{a_i = j, a_i \in S}|$ be the number of elements of size $j$ in $S$. Let $APT = \{(i_2(S),...,i_a(S))\}_{S\in AP}$ be the set of all unique tuples $(i_2(S),...,i_a(S))$ such that there is some $S\in AP$ with these size-counts values. Recall the notation $\displaystyle x^+ =\max(x,0)$. For any tuple $I = (i_2,...,i_a) \in APT$, there are $(\prod\limits_{j=2}^a \binom{r_j}{i_j})\binom{m}{(T - \sum\limits_{j=2}^a j \cdot i_j)^+}$ all-pivotal sets $S$ in $AP$ such that $(i_2(S),...,i_a(S)) = I$. This yields overall
$$|AP| = \sum\limits_{I\in APT} \big(\prod\limits_{j=2}^a \binom{r_j}{i_j}\big)\binom{m}{(T - \sum\limits_{j=2}^a j \cdot i_j)^+}.$$

Also for any $S\in AP$ and its corresponding $I\in APT$, it holds that
$$|S\cap A| = \sum\limits_{j=2}^a i_j, |S| = (T - \sum\limits_{j=2}^a j \cdot i_j)^+ + \sum\limits_{j=2}^a i_j.$$

We can write the sum
\begin{equation}
\label{DP_Ineq1}
\begin{split}
    & \sum\limits_{i=1}^r \rho_{a_i}(A,m,T) = \sum\limits_{i=1}^r \mathbb{E}_{S\sim UNI(AP)} \left[\frac{\mathbbm{1}_{a_i \in S}}{|S|}\right] = \frac{1}{|AP|} \sum\limits_{S\in AP} \frac{|S\cap A|}{|S|} = \\
& \frac{\sum\limits_{I\in APT} \big(\prod\limits_{j=2}^a \binom{r_j}{i_j}\big)\binom{m}{(T - \sum\limits_{j=2}^a j \cdot i_j)^+} \frac{\sum\limits_{j=2}^a i_j }{(T - \sum\limits_{j=2}^a j \cdot i_j)^+ + \sum\limits_{j=2}^a i_j } }{\sum\limits_{I\in APT} \big(\prod\limits_{j=2}^a \binom{r_j}{i_j}\big)\binom{m}{(T - \sum\limits_{j=2}^a j \cdot i_j)^+}}.
    \end{split}
    \end{equation}

Claim~\ref{nominatorInequalities} of the appendix shows that if $T < \frac{m}{2}$, then 
\[
\begin{split}
    & \sum\limits_{I\in APT} \big(\prod\limits_{j=2}^a \binom{r_j}{i_j}\big)\binom{m}{(T - \sum\limits_{j=2}^a j \cdot i_j)^+} \frac{\sum\limits_{j=2}^a i_j }{(T - \sum\limits_{j=2}^a j \cdot i_j)^+ + \sum\limits_{j=2}^a i_j } \leq \\
    & \frac{r}{m+1} \sum\limits_{\substack{I \in APT \\ \sum\limits_{j=2}^a j\cdot i_j < T }} \big(\prod\limits_{j=2}^a \binom{r_j}{i_j}\big) \binom{m + 1}{(T - \sum\limits_{j=2}^a j\cdot i_j - 2)^+ + 1}.
    \end{split}
\]
Then, Claim~\ref{nomDenomRatio} of the appendix states that for any $I\in APT$ with $\sum\limits_{j=2}^a j\cdot i_j < T < \frac{m}{2}$, it holds that $\displaystyle \binom{m + 1}{(T - \sum\limits_{j=2}^a j\cdot i_j - 2)^+ + 1} \leq 2\binom{m}{(T - \sum\limits_{j=2}^a j \cdot i_j)^+}$. By the two claims and Eq.~\ref{DP_Ineq1}, we have

\[
\begin{split}
& \sum\limits_{i=1}^r \rho_{a_i}(A,m,T) \stackrel{(Eq.~\ref{DP_Ineq1})}{=} \frac{\sum\limits_{I\in APT} \big(\prod\limits_{j=2}^a \binom{r_j}{i_j}\big)\binom{m}{(T - \sum\limits_{j=2}^a j \cdot i_j)^+} \frac{\sum\limits_{j=2}^a i_j }{(T - \sum\limits_{j=2}^a j \cdot i_j)^+ + \sum\limits_{j=2}^a i_j } }{\sum\limits_{I\in APT} \big(\prod\limits_{j=2}^a \binom{r_j}{i_j}\big)\binom{m}{(T - \sum\limits_{j=2}^a j \cdot i_j)^+}} \stackrel{(\ref{nominatorInequalities})}{\leq} \\
& \frac{r}{m+1} \cdot \frac{1}{\sum\limits_{I\in APT} \big(\prod\limits_{j=2}^a \binom{r_j}{i_j}\big)\binom{m}{(T - \sum\limits_{j=2}^a j \cdot i_j)^+}} \sum\limits_{\substack{I \in APT \\ \sum\limits_{j=2}^a j\cdot i_j < T }} \big(\prod\limits_{j=2}^a \binom{r_j}{i_j}\big) \binom{m + 1}{(T - \sum\limits_{j=2}^a j\cdot i_j - 2)^+ + 1} \stackrel{(\ref{nomDenomRatio})}{\leq} \\
& 2\frac{r}{m+1} \cdot \frac{1}{\sum\limits_{I\in APT} \big(\prod\limits_{j=2}^a \binom{r_j}{i_j}\big)\binom{m}{(T - \sum\limits_{j=2}^a j \cdot i_j)^+}} \sum\limits_{\substack{I \in APT \\ \sum\limits_{j=2}^a j\cdot i_j < T }} \big(\prod\limits_{j=2}^a \binom{r_j}{i_j}\big) \binom{m}{(T - \sum\limits_{j=2}^a j \cdot i_j)^+} \leq 2\frac{r}{m+1}
\end{split}
\]

We can then follow up with
\[
\begin{split}
    & \sum\limits_{i=1}^r \rho_{a_i}(A,m,T) \leq 2\frac{r}{m+1} = \frac{m+\sum\limits_{j=1}^r a_j}{m+1}\cdot \frac{2r}{m+\sum\limits_{j=1}^r a_j}\leq 2 \cdot \frac{\sum\limits_{j=1}^r a_j}{m + \sum\limits_{j=1}^r a_j} = 2P(A,m,T) ,
    \end{split}
    \]
with the inequality going from the first to second line due to $\displaystyle \sum\limits_{j=1}^r a_j < m$ (from the lemma assumptions) and $\displaystyle r \leq \frac{\sum\limits_{j=1}^r a_j}{2}$ (big players are at least of size 2). 
\end{proof}

\section{Technical Claims for the Deegan-Packel Upper Bound}
\begin{claim}
\label{nominatorInequalities}
If $T < \frac{m}{2}$, then 
\[
\begin{split}
    & \sum\limits_{I\in APT} \bigg(\prod\limits_{j=2}^a \binom{r_j}{i_j}\bigg)\binom{m}{(T - \sum\limits_{j=2}^a j \cdot i_j)^+} \frac{\sum\limits_{j=2}^a i_j }{(T - \sum\limits_{j=2}^a j \cdot i_j)^+ + \sum\limits_{j=2}^a i_j } \leq \\
    & \frac{r}{m+1} \sum\limits_{\substack{I \in APT \\ \sum\limits_{j=2}^a j\cdot i_j < T }} \bigg(\prod\limits_{j=2}^a \binom{r_j}{i_j}\bigg) \binom{m + 1}{(T - \sum\limits_{j=2}^a j\cdot i_j - 2)^+ + 1}.
    \end{split}
    \]
\end{claim}
\begin{proof}
We begin with the following inequalities: 
\begin{equation}
\label{nominator1}
\begin{split}
    & \sum\limits_{I\in APT} \bigg(\prod\limits_{j=2}^a \binom{r_j}{i_j}\bigg)\binom{m}{(T - \sum\limits_{j=2}^a j \cdot i_j)^+} \frac{\sum\limits_{j=2}^a i_j }{(T - \sum\limits_{j=2}^a j \cdot i_j)^+ + \sum\limits_{j=2}^a i_j } \stackrel{(1)}{=} \\
    & \sum\limits_{\substack{I\in APT \\ \sum\limits_{j=2}^a i_j \geq 1}} \bigg(\prod\limits_{j=2}^a \binom{r_j}{i_j}\bigg)\binom{m}{(T - \sum\limits_{j=2}^a j \cdot i_j)^+} \frac{\sum\limits_{j=2}^a i_j }{(T - \sum\limits_{j=2}^a j \cdot i_j)^+ + \sum\limits_{j=2}^a i_j } \stackrel{(2)}{\leq} \\
    & \sum\limits_{\substack{I\in APT \\ \sum\limits_{j=2}^a i_j \geq 1}} \bigg(\prod\limits_{j=2}^a \binom{r_j}{i_j}\bigg)\binom{m}{(T - \sum\limits_{j=2}^a j \cdot i_j)^+} \frac{\sum\limits_{j=2}^a i_j }{(T - \sum\limits_{j=2}^a j \cdot i_j)^+ + 1 },
    \end{split}
   \end{equation}
   
   where the transitions hold due to that:
   
   \begin{enumerate}
       \item Any all-pivotal size-counts tuple $I\in APT$ with $\sum\limits_{j=2}^a i_j = 0$ (and there is exactly one like that, which corresponds to the case where the all-pivotal set is comprised only of small players), will not contribute to the sum due to the $\sum\limits_{j=2}^a i_j$ expression in the numerator of the summand. 
   
   \item Since we only sum over $I\in APT$ with $\sum_{j=2}^a i_j \geq 1$, this condition is true for all summands, and so we can replace this expression in the inequality with 1 and draw the inequality. 

   \end{enumerate}

    Now notice that for each element $I$, we can rewrite
    \begin{equation}
    \label{elementWiseEquations}
            \begin{split}
        & \bigg(\prod\limits_{j=2}^a \binom{r_j}{i_j}\bigg)\binom{m}{(T - \sum\limits_{j=2}^a j \cdot i_j)^+} \frac{\sum\limits_{j=2}^a i_j }{(T - \sum\limits_{j=2}^a j \cdot i_j)^+ + 1 } \stackrel{(1)}{=} \\
        & \sum\limits_{2\leq k \leq a} \bigg(\prod\limits_{j=2}^a \binom{r_j}{i_j}\bigg)\binom{m}{(T - \sum\limits_{j=2}^a j \cdot i_j)^+} \frac{i_k }{(T - \sum\limits_{j=2}^a j \cdot i_j)^+ + 1 } \stackrel{(2)}{=} \\
        & \sum\limits_{\substack{2\leq k \leq a \\ i_k \neq 0}} \binom{r_k}{i_k} i_k \bigg(\prod\limits_{2\leq j \neq k \leq a} \binom{r_j}{i_j}\bigg)\binom{m}{(T - \sum\limits_{2\leq j\neq k \leq a}j \cdot i_j - k \cdot i_k )^+} \frac{1}{(T - \sum\limits_{2\leq j\neq k \leq a} j \cdot i_j - ki_k )^+ + 1 } \stackrel{(3)}{=} \\
        & \sum\limits_{\substack{2\leq k \leq a \\ i_k \neq 0}} \binom{r_k - 1}{i_k - 1} r_k \bigg(\prod\limits_{2\leq j \neq k \leq a} \binom{r_j}{i_j}\bigg)\binom{m + 1}{(T - \sum\limits_{2\leq j\neq k \leq a}j \cdot i_j - k \cdot i_k )^+ + 1} \frac{1}{m + 1}, 
    \end{split}
    \end{equation}

    with transitions due to:
    \begin{enumerate}
        \item Moving summation outside.
        \item All summands with $i_k = 0$ don't contribute to the sum, so we can limit the summation to $k$ values with $i_k \neq 0$. We can also isolate the $k$ terms in all the products and sub-sums of the summands. 
        
        \item In general for $a\geq b \geq 1$, it holds that $\binom{a - 1}{b - 1} \cdot a = \binom{a}{b} b$, by just moving terms in and out the binomial coefficient. We use this identity twice. 
    \end{enumerate}

    Now we may continue with 
    \begin{align*}
    &  \sum\limits_{\substack{I\in APT \\ \sum\limits_{j=2}^a i_j \geq 1}} \bigg(\prod\limits_{j=2}^a \binom{r_j}{i_j}\bigg)\binom{m}{(T - \sum\limits_{j=2}^a j \cdot i_j)^+} \frac{\sum\limits_{j=2}^a i_j }{(T - \sum\limits_{j=2}^a j \cdot i_j)^+ + 1 } \stackrel{(Eq.~\ref{elementWiseEquations})}= \\
    & \sum\limits_{\substack{I \in APT \\ \sum\limits_{j=2}^a i_j \geq 1}} \sum\limits_{\substack{2\leq k \leq a \\ i_k \neq 0}}  \binom{r_k - 1}{i_k - 1}r_k \bigg(\prod\limits_{2\leq j\neq k\leq a} \binom{r_j}{i_j}\bigg) \binom{m + 1}{(T - \sum\limits_{2\leq j\neq k \leq a}j \cdot i_j - k \cdot i_k )^+ + 1} \frac{1}{m+1} \stackrel{(1)}{=} \\
    & \sum\limits_{k=2}^{a} \sum\limits_{l_k=1}^{r_k} \sum\limits_{\substack{I \in APT \\ i_k = l_k}} \binom{r_k - 1}{i_k - 1} r_k  \bigg(\prod\limits_{2\leq j\neq k\leq a} \binom{r_j}{i_j}\bigg) \binom{m + 1}{(T - \sum\limits_{2\leq j\neq k \leq a}j \cdot i_j - k \cdot i_k )^+ + 1}\frac{1}{m+1} \stackrel{(2)}{=} \\
    & \frac{1}{m+1} \sum\limits_{k=2}^{a} r_k \sum\limits_{l_k=0}^{r_k-1} \sum\limits_{\substack{I \in APT \\ i_k = l_k + 1}}   \binom{r_k - 1}{l_k} \bigg(\prod\limits_{2\leq j\neq k\leq a} \binom{r_j}{i_j}\bigg) \binom{m + 1}{(T - \sum\limits_{2\leq j\neq k \leq a}j \cdot i_j - k \cdot l_k - k)^+ + 1} \stackrel{(3)}{\leq} \\
    & \frac{1}{m+1} \sum\limits_{k=2}^{a} r_k \sum\limits_{l_k=0}^{r_k-1} \sum\limits_{\substack{I \in APT \\ i_k = l_k + 1}} \binom{r_k}{l_k} \bigg(\prod\limits_{2\leq j\neq k\leq a} \binom{r_j}{i_j}\bigg) \binom{m + 1}{(T - \sum\limits_{2\leq j\neq k \leq a}j \cdot i_j - k \cdot l_k - k)^+ + 1} \stackrel{(4)}{\leq} \\
      & \frac{1}{m+1} \sum\limits_{k=2}^{a} r_k \sum\limits_{l_k=0}^{r_k-1} \sum\limits_{\substack{I \in APT \\ i_k = l_k + 1}} \binom{r_k}{l_k} \bigg(\prod\limits_{2\leq j\neq k\leq a} \binom{r_j}{i_j}\bigg) \binom{m + 1}{(T - \sum\limits_{2\leq j\neq k \leq a}j \cdot i_j - k \cdot l_k - 2)^+ + 1} \stackrel{(5)}{=} \\
      & \frac{1}{m+1} \sum\limits_{k=2}^{a} r_k \sum\limits_{\substack{I \in APT \\ \sum\limits_{j=2}^a j\cdot i_j < T }} \binom{r_k}{i_k} \bigg(\prod\limits_{2\leq j\neq k\leq a} \binom{r_j}{i_j}\bigg) \binom{m + 1}{(T - \sum\limits_{2\leq j\neq k \leq a}j \cdot i_j - k \cdot i_k - 2)^+ + 1} \stackrel{(6)}{=} \\
& \frac{1}{m+1} \sum\limits_{k=2}^{a} r_k \sum\limits_{\substack{I \in APT \\ \sum\limits_{j=2}^a j\cdot i_j < T }} \bigg(\prod\limits_{j=2}^a \binom{r_j}{i_j}\bigg) \binom{m + 1}{(T - \sum\limits_{j=2}^a j\cdot i_j - 2)^+ + 1} \stackrel{(7)}{=} \\
    & \frac{r}{m+1} \sum\limits_{\substack{I \in APT \\ \sum\limits_{j=2}^a j\cdot i_j < T }} \bigg(\prod\limits_{j=2}^a \binom{r_j}{i_j}\bigg) \binom{m + 1}{(T - \sum\limits_{j=2}^a j\cdot i_j - 2)^+ + 1}.
    \end{align*}
    
    with the transitions due to:
    \begin{enumerate}
        \item First, notice the summands do not change. We switch the order of summation, and separate the iteration over sets in $APT$, to first consider all possible non-zero values of $i_k$, namely $l_k$ between 1 and $r_k$. 
        
        \item We move the $\frac{1}{m+1}$ expression outside of all sums since it's not dependent on any sum parameter. We move the $r_k$ expression to the first sum since it only depends on $k$. We move $l_k$ to run from 0 to $r_k - 1$ instead of $1$ to $r_k$, and we make the substitution $i_k = l_k + 1$ in the summand.  
        
        \item Binomial coefficients are monotone increasing in the first argument, and so $\binom{r_k - 1}{l_k} \leq \binom{r_k}{l_k}$ for each summand.
        
        \item For a binomial coefficient $\binom{a}{b}$ with $b < \frac{a}{2}$, it is monotone increasing in the second argument. Since $T < \frac{m}{2}$, it holds in our case, and since $k\geq 2$, the second argument indeed non-decreasing by the substitution $k \rightarrow 2$. 
        
        \item Whenever $0\leq l_k \leq r_k - 1$ and $I\in APT$ has $i_k = l_k + 1$, we have $\displaystyle \sum\limits_{2\leq j\neq k\leq a} j \cdot i_j + k \cdot i_k = \sum\limits_{2\leq j\neq k\leq a} j \cdot i_j + k \cdot l_k + k < T + k$, otherwise some player of size $k$ would not be pivotal, and $I$ corresponds to an all-pivotal set. This is equivalent to $\displaystyle \sum\limits_{2\leq j\neq k\leq a} j \cdot i_j + k \cdot l_k < T $, so if we sum over all elements $I \in APT$ that have $\displaystyle \sum\limits_{j = 2}^a j \cdot i_j < T$ (Notice that for ease of notation we make the substitution $l_k = i_k$, we make the same substitution in the summand - it is ok since $i_k = l_k + 1$ is not in our summation assumptions anymore). 
        
        \item Isolating the $k$ element in the sub-sums and products is no longer required. 
        
        \item By definition $\displaystyle r = \sum\limits_{k=2}^a r_k$. 

    \end{enumerate}

\end{proof}
\begin{claim}
\label{nomDenomRatio}
For any $I\in APT$ with $\sum\limits_{j=2}^a j\cdot i_j < T < \frac{m}{2}$, it holds that $\binom{m + 1}{(T - \sum\limits_{j=2}^a j\cdot i_j - 2)^+ + 1} \leq 2\binom{m}{(T - \sum\limits_{j=2}^a j \cdot i_j)^+}$. 
\end{claim}

\begin{proof}
Let $I\in APT$ be some tuple of all-pivotal size-counts, with $\sum\limits_{j=2}^a j\cdot i_j < T$. 

We separate to two cases: 

\begin{enumerate}
\item If $\sum\limits_{j=2}^a j\cdot i_j = T - 1$, then $$\binom{m + 1}{(T - \sum\limits_{j=2}^a j\cdot i_j - 2)^+ + 1} = \binom{m+1}{1} = m+1 \leq 2m = 2\binom{m}{1} = 2\binom{m}{(T - \sum\limits_{j=2}^a j \cdot i_j)^+}. $$

\item If $\sum\limits_{j=2}^a j\cdot i_j < T - 1$, then by our assumption $T < \frac{m}{2}$ and so 
\begin{equation}
\label{simpleInequalityForClaim}
T - \sum\limits_{j=2}^a j\cdot i_j \leq m+1 - (T - \sum\limits_{j=2}^a j\cdot i_j) 
\end{equation}

which we can follow with 
\begin{equation}
\label{mainTechnicalTransitionForClaim}
\begin{split}
    &  \frac{\displaystyle \binom{m+1}{T -  \sum\limits_{j=2}^a j\cdot i_j - 1}}{ \displaystyle \binom{m}{T - \sum\limits_{j=2}^a j \cdot i_j}} = \frac{\displaystyle (m+1)(T - \sum\limits_{j=2}^a j \cdot i_j)}{\displaystyle (m - T + \sum\limits_{j=2}^a j\cdot i_j + 1)(m - T + \sum\limits_{j=2}^a j\cdot i_j + 2)} \stackrel{(Eq.~\ref{simpleInequalityForClaim})}{\leq} \\
    & \frac{m+1}{m - T + \displaystyle \sum\limits_{j=2}^a j\cdot i_j + 2} \leq \frac{m+1}{ \frac{m+1}{2}} = 2.
\end{split}
\end{equation}

Thus

\[
\begin{split}
    & \binom{m + 1}{(T - \sum\limits_{j=2}^a j\cdot i_j - 2)^+ + 1} = \binom{m + 1}{T - \sum\limits_{j=2}^a j\cdot i_j - 1} \stackrel{(Eq.~\ref{mainTechnicalTransitionForClaim})}{\leq} 2\binom{m}{T - \sum\limits_{j=2}^a j \cdot i_j} = 2\binom{m}{(T - \sum\limits_{j=2}^a j \cdot i_j)^+}.
\end{split}
\]

\end{enumerate}
\end{proof}


\section{Second Recursion Lemma and Duality}
\label{recursionAppendix}
\begin{lemma}
\label{recursionLemma2}
The following recursion holds for $\phi_1$

\[
\begin{split}
 & \phi_1(A,m,T) = \\
 & \begin{dcases}
  \frac{1}{m+r} &  T = m + \sum\limits_{j=1}^r a_j \\
 \frac{1}{m+r} (\sum\limits_{\substack{1\leq i \leq r \\ m + \sum\limits_{1\leq j \neq i \leq r} a_j \geq T}} \phi_1(A\!\setminus\!\{a_i\},m,T) + (m\!-\!1) \phi_1(A,m\!-\!1,T)) & T < m + \sum\limits_{j=1}^r a_j
 \end{dcases}.
\end{split}
\]

\end{lemma}
\begin{proof}

Fix some small player, for which we measure the expected number of permutations where it is pivotal. The recursion is done by conditional expectation on the identity of the last player in the permutation. If $T=m+\sum\limits_{j=1}^r a_j$, then if and only if our fixed player is last, is it pivotal, which happens in probability $\frac{1}{m+r}$. For $T<m+\sum\limits_{j=1}^r a_j$, if the last player in the permutation is $a_i$ and $m + \sum\limits_{1\leq j \neq i \leq r} a_j < T$, then our fixed player is not pivotal. If it is some big player $i$ with $m + \sum\limits_{1\leq j \neq i \leq r} a_j \geq T$, then the problem reduces to the weighted voting game with parameters $(A \setminus \{a_i\}, m, T)$. Similarly if a small player different than our fixed player is last (for which there are $m-1$ alternatives), the parameters are $(A, m-1, T)$. If our fixed player is last, it will not be pivotal. All above events are with a uniform probability of $\frac{1}{m+r}$. 
\end{proof}

\noindent
{\bf Remark. } The recursions in Lemma~\ref{recursionLemma} and Lemma~\ref{recursionLemma2} are dual and can be obtained from one another using the following lemma:

\begin{lemma}
The following duality holds for $\Phi$

$$\Phi_1(A,m,T) = \Phi_1(A,m, m + \sum\limits_{j=1}^r a_j - T + 1). $$
\end{lemma}
The proof of this lemma is simple and therefore omitted. Essentially, the lemma follows since for every permutation, one may iterate over the set of players to pass the threshold 'from the left' or 'from the right'.

\bibliography{ref.bib}
\bibliographystyle{theapa}

\end{document}